\documentclass[12pt]{article}
\usepackage[final]{epsfig}
\usepackage{amssymb}
\usepackage{latexsym}
\usepackage{wrapfig}
\usepackage{amsmath}
\usepackage{natbib} 
\usepackage{booktabs}
\usepackage{threeparttable}
\usepackage{xcolor}
\usepackage{amsthm}
\usepackage{comment}
\usepackage{graphicx, caption, subcaption}
\usepackage{authblk}
\usepackage{titlesec}
\allowdisplaybreaks

\newcommand{\bfe}{{\bf e}}

\newcommand{\bfn}{{\bf n}}

\newcommand{\bft}{{\bf t}}

\newcommand{\bfx}{{\bf x}}
\newcommand{\bfy}{{\bf y}}

\newcommand{\bfI}{{\bf I}}

\newcommand{\bfR}{{\bf R}}

\newcommand{\bfU}{{\bf U}}

\newcommand{\beq}{\begin{equation}}
\newcommand{\eeq}{\end{equation}}
\newcommand{\beqs}{\begin{eqnarray}}
\newcommand{\eeqs}{\end{eqnarray}}

\newcommand{\calI}{{\cal I}}

\newcommand{\calS}{{\cal S}}

\newtheorem{theorem}{Theorem}[section]

\newtheorem{lemma}{Lemma}[section]

\DeclareMathOperator{\sign}{sign}

\title{A continuum mechanics approach for the deformation of non-Euclidean origami generated by piecewise constant nematic director fields}
\author[1]{Linjuan Wang}
\author[2,*]{Fan Feng}
\affil[1]{School of Astronautics, Beihang University, Beijing 100191, China}
\affil[2]{School of Mechanics and Engineering Science, Peking University, Beijing 100871, China}
\affil[*]{Email: fanfeng@pku.edu.cn}

\begin{document}
	
	\maketitle 
	
	\begin{abstract}
We merge classical origami concepts with active actuation by designing origami patterns whose panels undergo prescribed metric changes. These metric changes render the system non-Euclidean, inducing non-zero Gaussian curvature at the vertices after actuation. Such patterns can be realized by programming piecewise constant director fields in liquid crystal elastomer (LCE) sheets. In this work, we address the geometric design of both compatible reference director patterns and their corresponding actuated configurations. On the reference configuration, we systematically construct director patterns that satisfy metric compatibility across interfaces. {We prove the existence and uniqueness of compatible director fields at a vertex for the generic case, up to orthogonal duals. The Gaussian curvature of the actuated vertex is computed based on the compatible director fields.}
On the actuated configuration, we develop a continuum mechanics framework to analyze the kinematics of non-Euclidean origami. In particular, we fully characterize the deformation spaces of three-fold and four-fold vertices and establish analytical relationships between their deformations and the director patterns. Building on these kinematic insights, we propose rational designs of large director patterns: one based on a quadrilateral tiling with alternating positive and negative actuated Gaussian curvature, and the other combining three-fold and four-fold vertices governed by a folding angle theorem. Remarkably, both designs achieve compatibility in both the reference and actuated states. {We also propose a design strategy for active metamaterials based on the periodic non-Euclidean origami. The active metamaterials can have two modes of motions by folding or stimulating.} We anticipate that our geometric framework will facilitate the design of non-Euclidean/active origami structures and broaden their application in active metamaterials, soft actuators, and robotic systems.
	\end{abstract}
	

\section{Introduction}

Shape programming has emerged as a fundamental topic across diverse fields, including robotics, metamaterials, soft electronics, and geometric design \cite{choi2019programming,filipov2015origami,dang2025kirigami}. Among various shape programming techniques, origami—the ancient art of paper folding—offers a promising platform for designing complex configurations through coordinated folding lines \cite{Levi2016Programming, lang1996computational, lang2011origami,misseroni2024origami}. Since paper is nearly inextensible, traditional origami involves isometric deformations that preserve distances within each panel. As a result, the surfaces remain developable, with zero Gaussian curvature, and the design space becomes highly constrained. The resulting nonlinearity and over-constrained geometry make many origami-related problems computationally challenging; for instance, determining rigid-foldability with optional creases is known to be NP-hard \cite{Akitaya_Demaine_Horiyama_Hull_Ku_Tachi_2020}.
Due to these challenges, much research focuses on special classes of origami, such as rigidly and flat-foldable quadrilateral mesh origami \cite{feng2020designs}, where inverse design frameworks have been developed \cite{dang2022inverse}. Additional work on rigid foldability includes \cite{tachi2009generalization,lang2018rigidly,izmestiev2017classification,hayakawa2024panel}, offering accurate predictions for actuation, though typically requiring external control mechanisms \cite{ai2024easy}.

An alternative approach to shape programming is metric design using responsive or functional materials. Systems inspired by biological growth \cite{huang}, swelling gels \cite{HONG20093282}, or phase-transforming liquid crystal elastomers (LCEs) \cite{warner2007liquid} can induce programmable metric changes (i.e., length changes) across a 2D domain. These metric changes directly influence the intrinsic geometry—such as Gaussian curvature—of the actuated surface \cite{klein, mostajeran2015curvature,mostajeran2016encoding}. Several inverse design frameworks have emerged, including patterning LCEs \cite{aharoni2014geometry,warner2018nematic,itay2019curved,griniasty2021multivalued}, designing active nets \cite{aharoni2018universal}, and controlling shape via pneumatic actuation \cite{siefert2019bio}. The metric-induced curvature can also enhance the mechanical strength of the structure \cite{guin2018layered}. These materials offer additional benefits such as remote control via light or magnetic fields \cite{kim2018printing}, and their softness and biocompatibility make them suitable for integration in biomedical systems \cite{sekitani2016ultraflexible}.
However, the main limitation of the metric design approach is that it prescribes only the in-plane geometry (the first fundamental form); the out-of-plane shape (second fundamental form) must still be determined via minimization of a bending energy functional. This often results in flexible systems with highly implicit, energetically-selected deformations.

\begin{figure}[h]
	\centering
	\includegraphics[width=\textwidth]{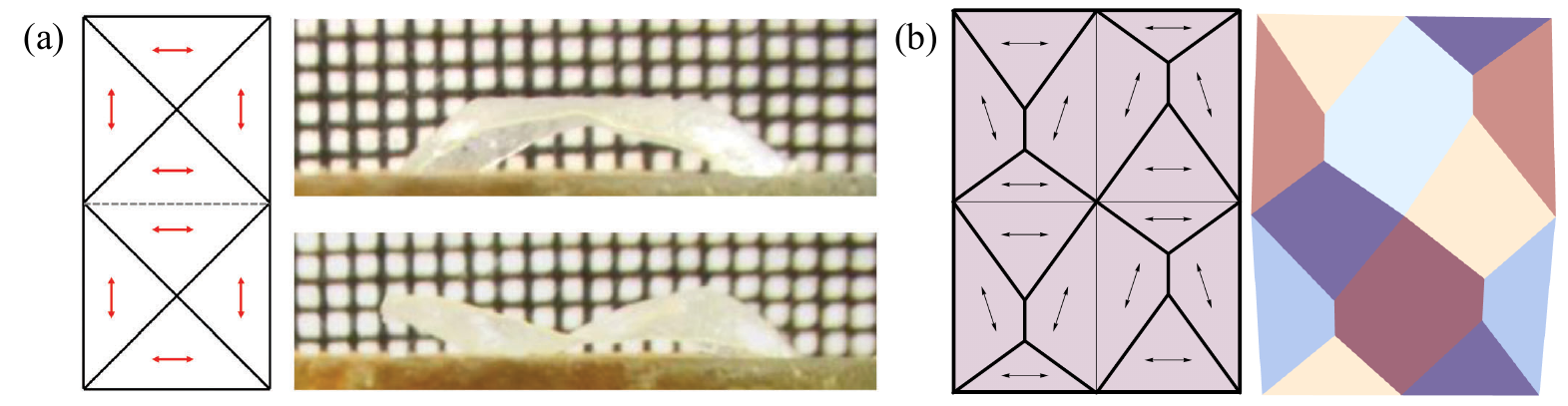}
	\caption{(a) An active origami adapted from \cite{plucinsky2018patterning}. (b) A generalized design composed of three-fold and four-fold vertices with no symmetry.}
	\label{fig:intro}
\end{figure}
Combining origami geometry with metric actuation seems to provide a new paradigm, potentially enabling shape-morphing systems that are both predictable, functionally programmable and remotely controllable. {In this work, the in-plane non-isometric deformation induces non-zero Gaussian curvature at the vertices in the deformed configuration, while the out-of-plane deformation is geometrically restricted by folding.} A seminal example is provided by Plucinsky et al. \cite{plucinsky2016programming,plucinsky2018patterning} (also shown in Fig.~\ref{fig:intro}), where origami panels are augmented with panels patterned by piecewise-constant LCE director fields. Upon actuation, the panels elongate along the director by a factor $\lambda_{\parallel}$ and contract transversely by $\lambda_{\perp}$. Although each panel remains flat, the actuation induces non-isometric deformations, giving rise to non-trivial Gaussian curvature at the vertices \cite{o2006elementary,duffy2020defective}. The resulting configuration often bifurcates into distinct “up” and “down” states due to the induced curvature. While the experiments validate theoretical predictions, their design space remains limited due to symmetry constraints, limiting applicability in scenarios requiring more complex actuation. We also find that the systematic study of active non-Euclidean origami is quite limited in the literature.
To enable more versatile functionality, it is crucial to develop a generalized design framework for active origami (also generically resulting in non-Euclidean origami \cite{waitukaitis2019non}) that can systematically handle director pattern design and accurately predict deformation behavior. 

In this paper, we focus on the geometric and kinematic aspects of such a framework.
We restrict our attention to a purely geometric regime, assuming that each panel undergoes a uniform, piecewise-constant metric change prescribed by the director field, and remains flat and rigid after actuation. Within this framework, two main challenges arise. First, the reference director pattern must be compatible, meaning the prescribed length changes must match across each crease. A simple case occurs when the crease is the bisector of adjacent directors or orthogonal to it—a condition referred to as metric compatibility \cite{feng2020evolving}. However, patterning compatible director fields becomes nontrivial when accounting for topological and geometric interactions among multiple creases. For example, it is not obvious that a compatible director fields at a vertex exists.
Second, the kinematics of the actuated configuration must be predicted accurately. Ideally, the kinematics of the actuated shape can be predicted directly from the reference director pattern. A key question remains: even when the reference director pattern is compatible, does the actuated configuration maintain compatibility?

We address these challenges by first analyzing the fundamental conditions for metric compatibility, proving the existence of the compatible director fields at a general $n$-fold vertex, and then studying the kinematics of actuated three-fold and four-fold intersections.
{By satisfying the metric compatibility condition, the length changes of the crease are identical from both sides. Under this condition, we prove that compatible director fields always exist at a $n$-fold vertex. They are characterized as the generic case with two determined director fields which are orthogonal to each other (orthogonal dulas) and the degenerate case with a continuous distribution of the director fields.}
We then generalize previous continuum mechanics approaches \cite{feng2020designs, zou2024kinematics} to derive closed-form kinematic equations for these non-Euclidean vertices. Specifically, we show that the actuated three-fold vertices are rigid with no free folding angle and are fully determined by the director configuration, while the four-fold vertices have one free folding angle, with the remaining angles following from a unified analytical formula.
Using these insights, we investigate multi-node interactions. For example, a pattern with two connected three-fold vertices reveals hidden degrees of freedom in director programming. For larger-scale patterns, we propose a marching algorithm to design quadrilateral director tilings with alternating positive and negative actuated Gaussian curvature. Additionally, we construct another class of compatible patterns composed of both three-fold and four-fold vertices (Fig.~\ref{fig:intro}(b)), governed by a folding angle theorem. Surprisingly, despite the complexity of these systems, the active origami patterns are compatible in both the reference and actuated configurations and offer substantial design flexibility. {Moreover, building on the periodic quadrilateral pattern, we propose a design for active metamaterials that can deform through folding and external stimulation. These active metamaterials have the potential to be coarse-grained to have a new shell/plate theory \cite{mcinerney2025coarse, li2025nonlinear, XU2024105832}.}

This paper is organized as follows. Section 2 introduces the deformation of director patterns and the conditions for metric compatibility. Section 3 studies the patterning of
compatible director fields at a general $n$-fold vertex.
Sections 4 and 5 analyze the kinematics of three-fold and four-fold vertices, respectively. Section 6 presents two design strategies for large-scale director fields and their actuated configurations. A design strategy for active metamaterials is also proposed.
Finally, Section 7 concludes the paper and discusses future research directions.

{\bf Notation.} $\{\bfe_1, \bfe_2, \bfe_3\}$ is the standard right-handed orthonormal basis for $\mathbb{R}^3$. $\bfR_{\bfe_3}(\theta)$ is a rotation tensor with  rotation axis $\bfe_3$ and  rotation angle $\theta$. The sector angle is denoted as $\alpha_i$ on the reference domain and $\tilde{\alpha}_i$ on the actuated domain respectively.  

\section{Deformation of constant director fields and metric compatibility condition} \label{sec:metric_compatibility}
We first consider the deformation induced by actuating a director pattern
\beq 
\bfn(\bfx) = \cos\theta(\bfx) \bfe_1 + \sin\theta(\bfx) \bfe_2,~ \bfx \in \mathbb{R}^2
\eeq
in the $\bfe_1,\bfe_2$ plane,
where $\bfn(\bfx)$ is a unit vector. Upon actuation, the pattern contracts along the director $\bfn$ by a factor $\lambda_{\parallel}$ and elongates along the perpendicular direction $\bfn_\perp$ by a factor $\lambda_{\perp}$. Here $\bfn_\perp$ is also called the orthogonal dual of $\bfn$ satisfying $\bfn_\perp \cdot \bfn =0$. Specifically, in liquid crystal elastomers/glasses, the  contraction factor $\lambda_{\parallel} = \lambda$ is $\sim 0.7$ and the elongation factor $\lambda_{\perp} = \lambda^{-\nu}$, with the {\it opto-thermal Poisson's ratio} $\nu=1/2$ for incompressible elastomers and as high as $\nu=2$ for glasses. Encoding the contraction and elongation, one can write the deformation gradient of the actuation as 
\beq
\bfU_{\bfn} = \lambda_{\parallel} \bfn \otimes \bfn + \lambda_{\perp} \bfn_\perp \otimes \bfn_\perp + \bfe_3 \otimes \bfe_3, \label{eq:gradient}
\eeq
which implies $\bfU_{\bfn} \bfn = \lambda_{\parallel} \bfn$ and $\bfU_{\bfn} \bfn_\perp = \lambda_{\perp} \bfn_\perp$.
Notice that $\bfU_{\bfn}$ is invariant under the transformation $\bfn \rightarrow -\bfn$, which means the director $\bfn$ is actually double-arrowed, and changing $\bfn$ to $-\bfn$ results in the same pattern and actuation. 
\begin{figure}[h]
	\centering
	\includegraphics[width=\textwidth]{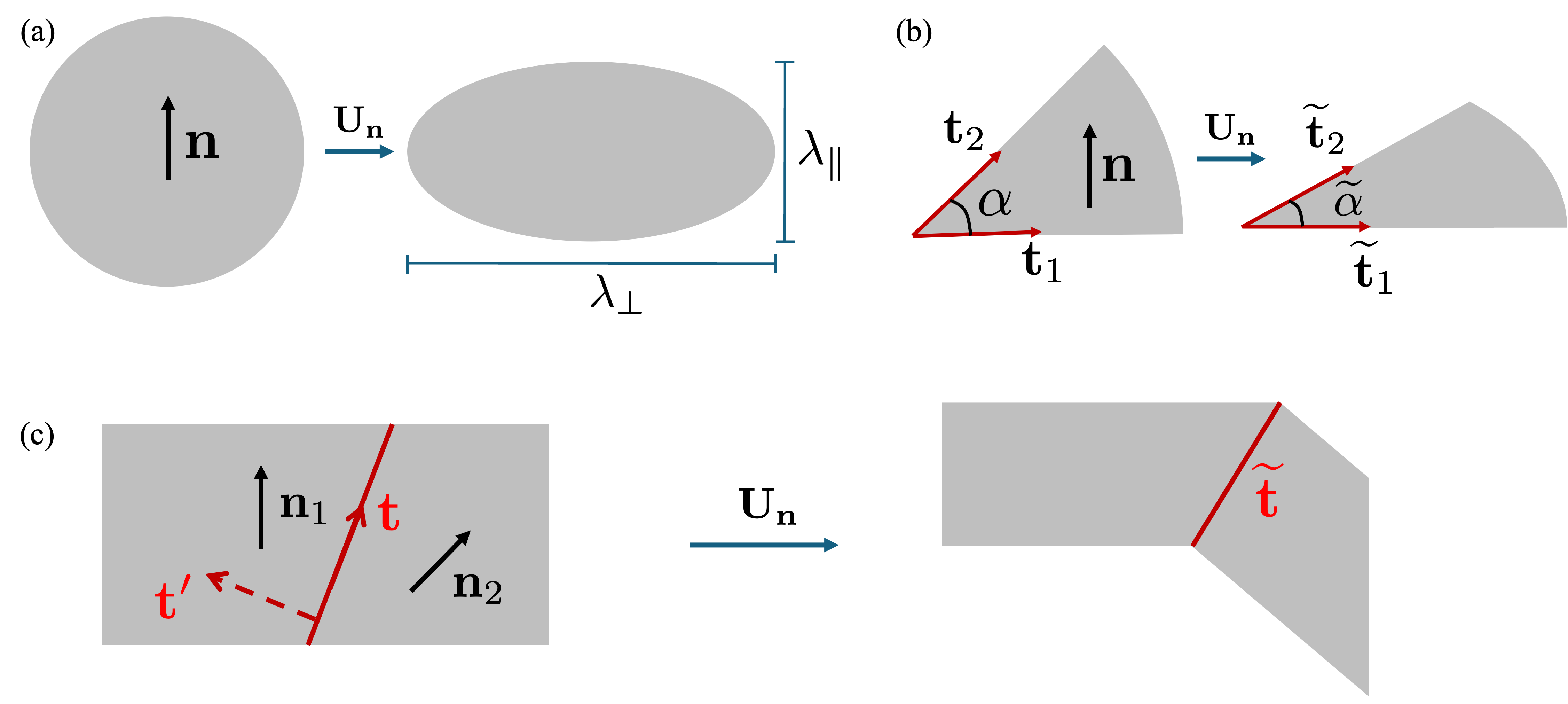}
	\caption{(a) A disc with constant director is deformed into an ellipse. (b) A sector angle $\alpha$ is deformed to $\tilde{\alpha}$. (c) A metric compatible interface with tangent $\bft$ (or $\bft'$) between two directors $\bfn_1$ and $\bfn_2$.}
	\label{fig:angle}
\end{figure}
The deformation of a constant director pattern is shown in Fig.~\ref{fig:angle}(a), where a unit disc is deformed into an ellipse with long axis $\lambda_\perp$ and short axis $\lambda_\parallel$ under $\bfU_{\bfn}$. Accordingly, the change of sector angles for a circular sector before and after actuation can be computed. The sector angle $\alpha\in(0,\pi)$ and two unit side vectors $\bft_1$, $\bft_2$ on the reference domain are illustrated in Fig.~\ref{fig:angle}(b). Under the deforamtion gradient $\bfU_{\bfn}$, the reference and deformed sector angles $\alpha$ and  $\tilde{\alpha}$ are given by
\beq
\alpha =\alpha(\bft_1, \bft_2) = \arccos(\bft_1 \cdot \bft_2), \quad
\tilde{\alpha} = \tilde{\alpha}(\bft_1, \bft_2, \bfn) = \arccos\left(\frac{\bfU_{\bfn} \bft_1 \cdot \bfU_{\bfn}\bft_2}{|\bfU_{\bfn} \bft_1| |\bfU_{\bfn} \bft_2|}\right) 
\label{eq:sector_angle}
\eeq
respectively. 

Next, we consider the metric-compatible interface with tangent $\bft$ between the two constant director fields $\bfn_1$ and $\bfn_2$ satisfying $\bfn_1 \nparallel \bfn_2$. The interface will serve as the crease of origami in the following sections. If we assume that the  interfaces are compatible before and after actuation (the lengths are identical from the two sides) and $\lambda_{\parallel,\perp}$ are uniform throughout the entire pattern, one can obtain the metric-compatibility condition
\beq
|\bfU_{\bfn_1} \bft| = |\bfU_{\bfn_2} \bft| \Leftrightarrow (\bfn_1 \cdot \bft)^2 = (\bfn_2 \cdot \bft)^2,  \label{eq:metric}
\eeq
which yields the unit tangent $\bft$ as
\beq
\bft =\frac{\sigma \bfn_1 + \tilde{\sigma} \bfn_2}{|\sigma \bfn_1 + \tilde{\sigma} \bfn_2|}, \quad \sigma, \tilde{\sigma} \in \{ \pm 1 \}. \label{eq:tangent}
\eeq
Figure~\ref{fig:angle}(c) shows an example with $\bft = (\bfn_1 + \bfn_2)/|\bfn_1 +\bfn_2|$ on the reference domain. By satisfying the condition (\ref{eq:tangent}), the deformed interfaces can be joined together with tangent $\tilde{\bft}$ by appropriate rigid motions. It should be noted that, given the two director fields, the metric-compatible interface is either a bisector or a perpendicular dual to the bisector ($\bft$ and $\bft'$ in Fig.~\ref{fig:angle}(c)).

{
\section{Patterning a compatible $k$-fold director vertex} \label{sec:nfold}
We first consider a general $k$-fold intersection in 2D, where $k$ distinct director fields $\bfn_1,\dots, \bfn_k$ ($\bfn_i \nparallel \bfn_j$ if $i\neq j$) are separated by $k$ creases $\bft_1,\dots,\bft_k$. At each crease, the metric compatibility is satisfied as $|\bfU_{\bfn_i} \bft_i| = |\bfU_{\bfn_{i-1}} \bft_i|$. In the following, we provide two patterning strategies, prescribing either the director fields or the creases. Then we compute the concentrated Gaussian curvature after actuation based on the 2D director pattern.

\subsection{Two patterning strategies}
Suppose that $k$ director fields $\bfn_1,\dots, \bfn_k$ ($\bfn_i \nparallel \bfn_j$ if $i\neq j$)  are prescribed. We design creases $\bft_i$ between $\bfn_i$ and $\bfn_{i-1}$ which satisfies the metric compatibility. As discussed above, such $\bft_i$ can be written as Eq.~(\ref{eq:tangent}), meaning that $\bft_i$ can bisect the two adjacent directors or be perpendicular to the bisector. However, the branches $\sigma$ and $\tilde{\sigma}$ cannot be arbitrarily chosen.  As shown in Fig.~\ref{fig:nfold}, the creases must satisfy the 
topological condition
\beq
\sign[(\bft_i \times \bft_{i+1}) \cdot \bfe_3] = +1,\quad {\text{for } i=1,2,...,k} \label{eq:topolotial_condition}
\eeq
to preserve the counterclockwise sequence of $\bft_i$. The topological condition itself is subtle to be satisfied, as the branches $\sigma$ and $\tilde{\sigma}$ can not be explicitly implied by this condition. We have to substitute specific $\sigma$ and $\tilde{\sigma}$ and check the condition.
Thus, a more realistic strategy is to prescribe the creases and then distribute the director fields.

Now suppose that $k$ creases are given by $\bft_{i+1} = \bfR_{\bfe_3}(\alpha_1+\dots+\alpha_i) \bfe_1$ and $\bft_1=\bfe_1$, where $\alpha_i$ satisfies the topological condition $0<\alpha_i<\pi$ and $\alpha_1+\alpha_2+\dots+\alpha_k = 2\pi$. We first assume that a compatible director field $\bfn_i,i=1,2,\dots,k$ exists. An orthogonal dual to this director field is defined as a new field by replacing $\bfn_i$ with $\bfn_i^{\perp}=\bfR_{\bfe_3}(\pi/2) \bfn_i$ in the corresponding domains, as shown in Fig.~\ref{fig:nfold}(a). The following Lemma implies that a compatible director field's orthogonal dual is also compatible.

\begin{lemma}
 Suppose a $k$-fold  director field is compatible at each crease ($|\bft_i \cdot \bfn_{i-1}| = |\bft_i \cdot \bfn_i|$). Then the orthogonal dual is also compatible ($|\bft_i \cdot \bfR_{\bfe_3}(\pi/2)\bfn_{i-1}| = |\bft_i \cdot \bfR_{\bfe_3}(\pi/2)\bfn_i|$).    
\end{lemma}
\begin{proof}
    The proof is straightforward by projecting $\bfn_i,\bfn_{i-1}$ onto $\bft_i$ and $\bft_i^{\perp}$ as $\bfn_{i-1} = a_1 \bft_i + b_1 \bft_i^\perp, \bfn_{i} = a_2 \bft_i + b_2 \bft_i^\perp$. The compatibility $|\bft_i \cdot \bfn_{i-1}| = |\bft_i \cdot \bfn_i|$ yields $|a_1|=|a_2|$. Since $\bft_i$ and $\bft_{i-1}$ are both unit vectors, we then have $|b_1|=|b_2|$, which is the compatibility $|\bft_i \cdot \bfR_{\bfe_3}(\pi/2)\bfn_{i-1}| = |\bft_i \cdot \bfR_{\bfe_3}(\pi/2)\bfn_i|$ for the orthogonal dual. 
\end{proof}

\begin{figure}[h]
	\centering
	\includegraphics[width=\textwidth]{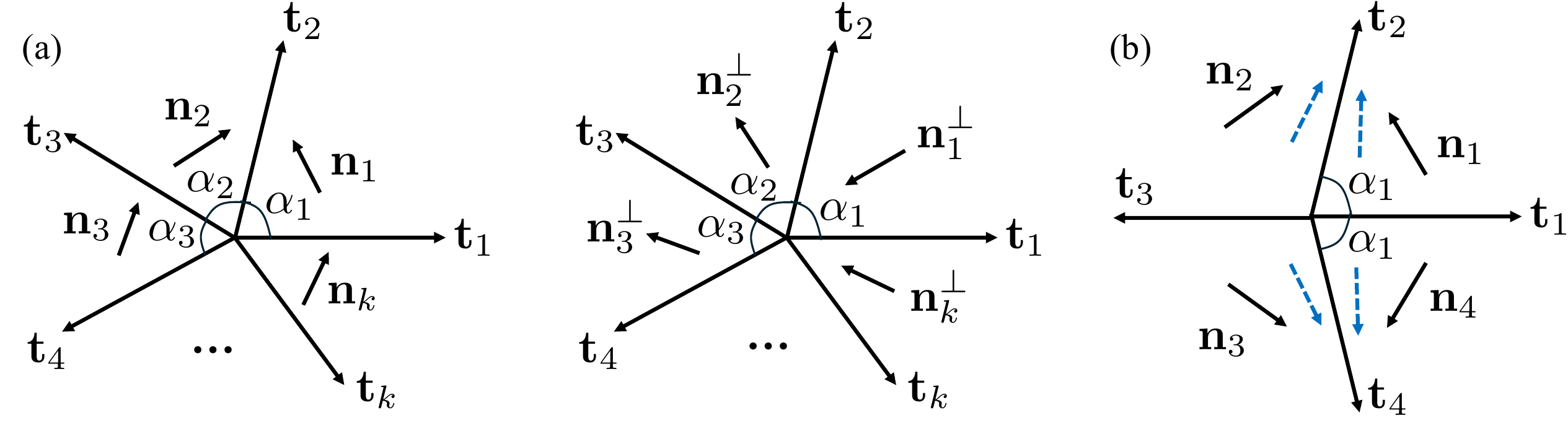}
	\caption{(a) An n-fold director field and its orthogonal dual for the generic case. The solution for $\bfn_i$ is unique up to orthogonal duals. (b) A degenerate case that has a continuous family of solutions for $\bfn_i$. }
	\label{fig:nfold}
\end{figure}
The following theorem implies the existence and uniqueness of the compatible director fields for a given set of creases (generic case), up to orthogonal duals.
\begin{theorem}
    Suppose $k$ creases are prescribed by $\bft_1=\bfe_1$ and $\bft_i = \bfR_{\bfe_3}(\alpha_1+\alpha_2+\dots + \alpha_{i-1}) \bfe_1, i=2,\dots,k$, where the sector angles satisfy $0 <\alpha_i<\pi$ and $\alpha_1+\dots + \alpha_k = 2\pi$. Then generically there are two and only two compatible director fields, which are the orthogonal duals to each other, except for the degenerate cases.
\end{theorem}
\begin{proof}
    Without loss of generality, we have $\bft_1 = \bfe_1$. Let $\bfn_1 = \bfR_{\bfe_3} (\theta) \bfe_1, \theta \in (0,\pi)$. To satisfy the metric compatibility $|\bft_i\cdot \bfn_{i-1}| = |\bft_i \cdot \bfn_i|$, we systematically pattern the directors as $\bfn_{i+1} = \bfn_i - 2 (\bfn_i \cdot \bft_{i+1}) \bft_{i+1}$ for $i=1,\dots,k-1$. This will make the compatibility condition $\bft_i\cdot \bfn_{i-1} = -\bft_i \cdot \bfn_i$ to be satisfied at all creases except $\bft_1$. Thus we may check $error=(\bfn_k \cdot \bft_1)^2 - (\bfn_1 \cdot \bft_1)^2$ for the last compatibility at $\bft_1$. Notice that $\bfn_1$ is a function of $\theta$ and $\bft_i$ are given. Thus, the other directors can be derived sequentially, resulting in a function $error(\theta)$. Substituting $\theta$ and $\alpha_i$, a direct calculation yields
    \begin{align}
        error(\theta) = \begin{cases}
            -\sin[2(\alpha_1 + \alpha_3 + \dots + \alpha_{k-1})] \sin[2(\alpha_1 + \alpha_3 + \dots + \alpha_{k-1}+\theta)] \quad \text{for even k} \\
             -\sin[2(\alpha_2 + \alpha_4 + \dots + \alpha_{k-1})] \sin[2(\alpha_2 + \alpha_4 + \dots + \alpha_{k-1}+\theta)] \quad \text{for odd k}
        \end{cases}.
    \end{align}
Notice that $\alpha_1 + \alpha_2 + \dots + \alpha_k = 2\pi$. Thus $error(\theta) \equiv 0$ if the sum of the even indexed terms or the odd indexed terms in $\alpha_i$ is $\pi/2,\pi$ or $3\pi/2$, and $\theta$ can be any value in $(0,\pi)$ (degenerate case in Fig.~\ref{fig:nfold}(b)). Otherwise, $error(\theta) = 0$ has two solutions for $\theta \in (0,\pi)$. Meanwhile, we need to check $\bfn_k \times \bfn_1 \neq 0$ to make sure $\bfn_k$ is not parallel to $\bfn_1$, and $\bft_1$ is a crease. Again, by direct calculation, we have 
\begin{align}
    \bfn_k \times \bfn_1 = \begin{cases}
        (0,0,-\sin[2(\alpha_1 + \alpha_3 + \dots +\alpha_{k-1}+\theta)]) \quad \text{for even k} \\
        (0,0,-\sin[2(\alpha_2 + \alpha_4 + \dots +\alpha_{k-1})]) \quad \text{for odd k} 
    \end{cases}.
\end{align}
Taking these two conditions together, we conclude that
\begin{enumerate}

    \item Generic case: $\alpha_2 + \alpha_4 + \dots \neq \pi/2, \pi$ or $3\pi/2$. If $k$ is odd, there are two compatible patterns for $\bfn_i$ (by satisfying $\sin[2(\alpha_2 + \alpha_4 + \dots + \alpha_{k-1}+\theta)]=0$) which are orthogonal duals to each other, since $error(\theta)$ and $error(\theta \pm \pi/2)$ vanish simultaneously (Fig.~\ref{fig:nfold}(a)). If $k$ is even, we eventually get $\bfn_k \parallel \bfn_1$ by patterning a compatible director filed, which means that $\bft_1$ is not a crease and the even $k$ case degenerates into the odd $k$ case.

    \item Degenerate case: $\alpha_2 + \alpha_4 + \dots =\pi/2, \pi$ or $3\pi/2$. If the sum of the even indexed (or equivalently odd indexed) sector angles is $\pi/2, \pi$ or $3\pi/2$, say $\alpha_2 + \alpha_4 + \dots = \pi/2, \pi$ or $3\pi/2$, and $k$ is even, we have $\theta \in (0,\pi)$ is always a solution since $\sin[2(\alpha_1 + \alpha_3 + \dots + \alpha_{k-1})]=0$. The number of compatible patterns $\bfn_i$ is infinite. If $k$ is odd and $\alpha_2 + \alpha_4 + \dots = \pi/2, \pi$ or $3\pi/2$, we have $\bfn_k \times \bfn_1 =0$ ($\bfn_k \parallel \bfn_1$) and the case is degenerated into the even $k$ case.

\end{enumerate}
The results can be conveniently confirmed by using an alternative Householder transformation (matrix) approach in mirror reflections.
\end{proof}

\subsection{Concentrated Gaussian curvature at the vertex}
In this section, we study the Gaussian curvature concentrated at the vertex after deformation for a compatible director field. 
As discussed above, we only need to consider two cases: the generic case with odd $k$ and the degenerate case with even $k$. The former case has two solutions for compatible director patterns which are orthogonal duals, and the latter case possesses a continuous variation of director patterns. 

For the generic case with odd $k$, we have two solutions for $error(\theta)=0$ which are $\theta = \theta_0:= -\alpha_2-\alpha_4-\dots-\alpha_{k-1}$ and $\theta = \theta_0+\pi/2$. Then we have $\bfn_1 = \bfR_{\bfe_3}(\theta) \bfe_1$ and other $\bfn_i$ can be computed successively by  $\bfn_{i+1} = \bfn_i - 2 (\bfn_i \cdot \bft_{i+1}) \bft_{i+1}$.
Since the director we study is piecewise constant, the actuated pattern possesses non-zero Gaussian curvature at the tip and zero Gaussian curvature elsewhere. According to previous work \cite{feng2020evolving}, the concentrated Gaussian curvature at the tip is $G= 2\pi - \sum_{i=1}^k \tilde{\alpha}_i$, where $\tilde{\alpha}_i$ is the actuated sector angle given by Eq.~(\ref{eq:sector_angle}). In the following, we use $\beta_i = \sum_{j=1}^{i} \alpha_j$ for convenience, so that the creases can be written as $\bft_i = \bfR_{\bfe_3} (\beta_i) \bfe_1$.  Substituting the creases $\bft_i$ and the directors $\bfn_i$ with $\theta_0=-\alpha_2-\alpha_4 - \dots - \alpha_{k-1}$, we have the actuated Gaussian curvature at the tip as a function of $\alpha_i$:
\beqs
G&=&G_{odd}(k,\alpha_1,\dots,\alpha_k) \nonumber\\  &:=&
2\pi - \sum_{i=1}^k \arccos \frac{(r^2+1) \cos \alpha_i + (r^2 - 1) \cos \tilde{\beta}_i}{\sqrt{[(r^2+1) + (r^2-1)\cos(\alpha_i  - \tilde{\beta}_i)] [(r^2+1) + (r^2-1)\cos(\alpha_i  +\tilde{\beta}_i)]}} \label{eq:Gauss}\\
&&\text{where } \tilde{\beta}_i = 2\sum_{j=1}^{i-2} (-1)^{j+1} \beta_j + (-1)^i \beta_{i-1} + (-1)^i \beta_i + 2 \sum_{j=i+1}^{k-1}(-1)^j \beta_j, ~\beta_i = \sum_{j=1}^{i} \alpha_j \text{ and } r= \frac{\lambda_{\parallel}}{\lambda_{\perp}}. \nonumber
\eeqs
\begin{figure}[!t]
	\centering
	\includegraphics[width=0.8\textwidth]{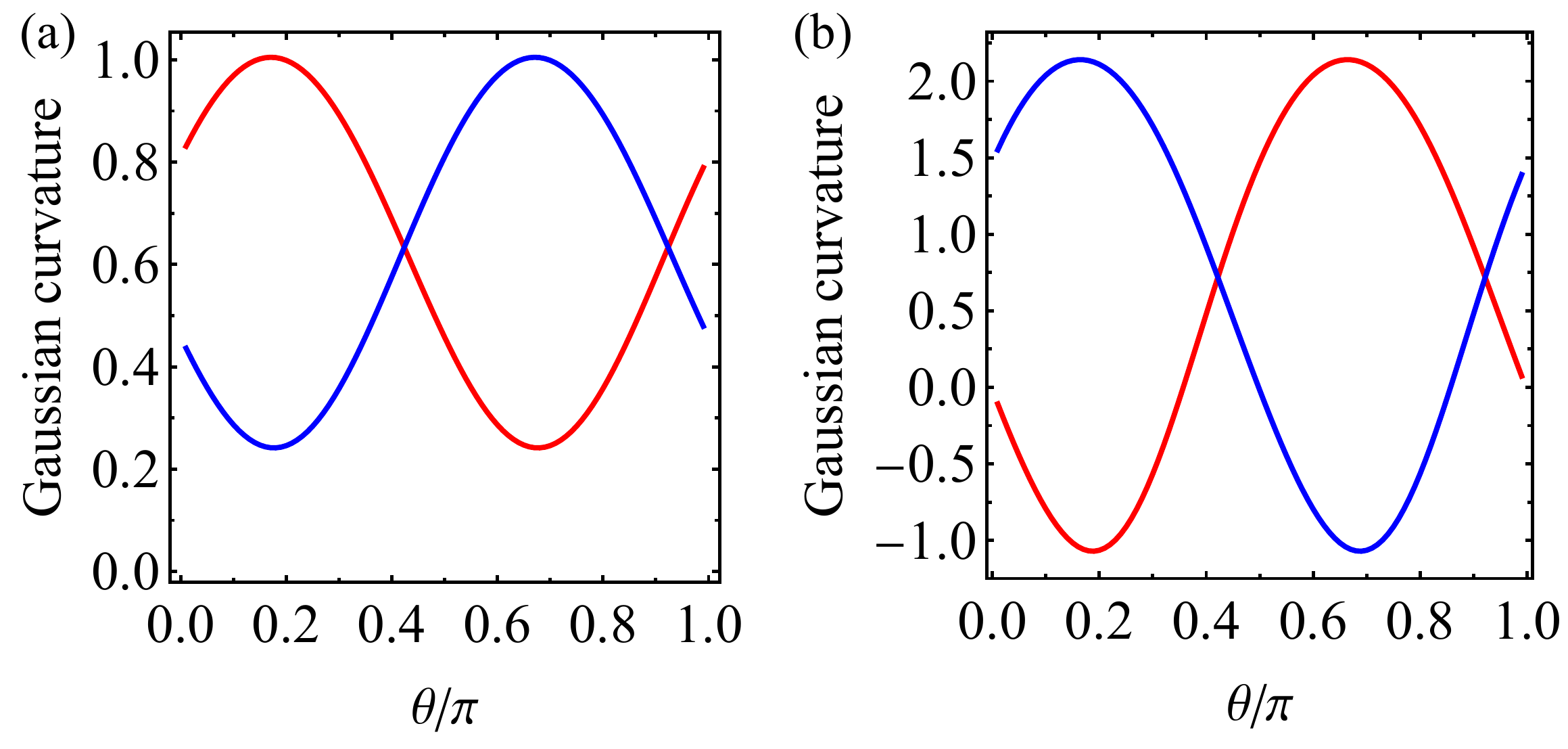}
	\caption{The Gaussian curvature of orthogonal duals (red and blue) for the degenerate case with (a) $k=6, r=0.9$, $ (\alpha_1,\alpha_2,\alpha_3,\alpha_4,\alpha_5)=2\pi(0.2,0.13,0.07,0.2,0.23)$ and (b) $k=6, r=2/3$, $ (\alpha_1,\alpha_2,\alpha_3,\alpha_4,\alpha_5)=2\pi(0.2,0.13,0.07,0.2,0.23)$. The Gaussian curvature of the degenerate case can be modulated within a wide range (positive or negative).}
	\label{fig:degenerate}
\end{figure}
Here if the index $j \notin \{1,2,\dots,k \}$, we set $\alpha_j =0$.
 Given the creases ($\alpha_i$) for the generic case, Eq.~(\ref{eq:Gauss}) provides an explicit formula for the Gaussian curvature of the tip upon actuation. The Gaussian curvature of the orthogonal dual can be computed by simply replacing $r$ with $1/r$. Since a concentric circle pattern and a radial director pattern generate positive and negative Gaussian curvature respectively for $r<1$ \cite{feng2022interfacial}, it is natural to ask whether the two orthogonal duals always generate positive and negative Gaussian curvature. Though not explicitly seen from (\ref{eq:Gauss}), the answer is ``No". A numerical test with $k=5, r=1.1$ and $(\alpha_1,\alpha_2,\alpha_3,\alpha_4) = 2\pi(0.2,0.1,0.05,0.55)$ gives $G=0.445$ and $G=0.807$ for the two orthogonal dulas, which are both positive. Another example is the degenerate case in Fig.~\ref{fig:degenerate}(a). We refer the reader to \cite{duffy2020defective} for the Gaussian curvature of a generic director defect.

For the degenerate case with even $k$, the solution $error(\theta)=0$ is a continuous range $\theta \in (0,\pi)$. Thus, we may expect that the actuated Gaussian curvature can also be modulated within a range. Suppose $k$ is even, $\alpha_2 + \alpha_4 + \dots =\pi/2, \pi$ or $3\pi/2$ and $\bfn_1=\bfR_{\bfe_3}(\theta)\bfe_1$ for $\theta \in (0,\pi)$. We may follow exactly the same process to compute a Gaussian curvature function $G=G_{even}(k,\alpha_1,\dots,\alpha_k)$ for even $k$. We omit the specific derivation and claim that $G_{even}(k,\alpha_1,\dots,\alpha_k)$ can be obtained from $G_{odd}(k+1,\alpha_1,\dots,\alpha_k,\alpha_{k+1})$ by adding an arbitrary $\alpha_{k+1}$, replacing $\tilde{\beta}_i$ with  $\tilde{\beta}_i + 2\theta_0 -2 \theta$ (now $\theta_0=-\alpha_2-\alpha_4 - \dots - \alpha_{k}$) and changing the upper limit of summation from $k+1$ to $k$ in $G_{odd}(k+1,\alpha_1,\dots,\alpha_k,\alpha_{k+1})$ . The formula has been numerically validated. Figure~\ref{fig:degenerate} shows that the Gaussian curvature of the degenerate case can be modulated within a wide range (positive or negative) for the orthogonal duals.

}

\section{Three-fold intersection} \label{sec:kinematics}
In this section, we study the kinematics of three-fold intersections with patterned piecewise constant director fields satisfying the metric compatibility condition at each interface. Our approach is shown in Fig.~\ref{fig:three_fold}(a). The reference pattern is metric-compatible at each interface. To compute the actuated configuration, the first step is to cut the reference pattern along the interface $\bft_1$ and then deform it by the spontaneous deformation gradient $\bfU_{\bfn_i}$ at each domain. This will form an intermediate configuration in 2D,  with a deficit or a surplus. The second step is to deform the intermediate configuration isometrically and stitch the boundaries together to form a non-Euclidean origami with Gaussian curvature at the tip. We use $\alpha_i$ and $\tilde{\alpha}_i$ to denote the sector angles before and after actuation, respectively.

\begin{figure}[ht]
	\centering
	\includegraphics[width=\textwidth]{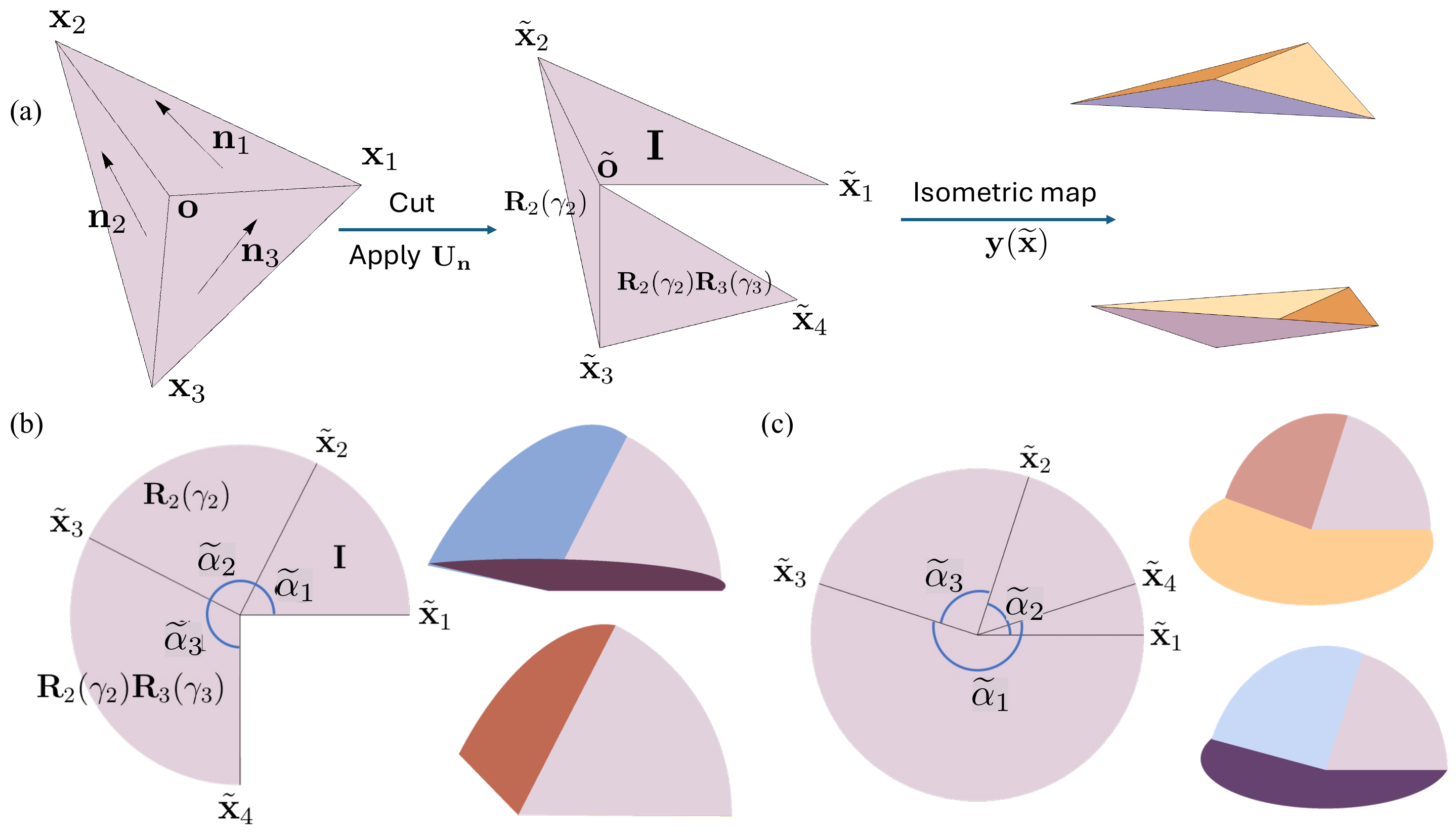}
	\caption{(a) A three-fold non-Euclidean origami constructed through cutting, deforming by $\bfU_{\bfn}$, and stitching. (b) A three-fold intersection with a deficit ($\tilde{\alpha}_1+\tilde{\alpha}_2+\tilde{\alpha}_3<2\pi$). (c) A three-fold intersection with a surplus ($\tilde{\alpha}_1+\tilde{\alpha}_2+\tilde{\alpha}_3>2\pi$).}
	\label{fig:three_fold}
\end{figure}

\subsection{Three-fold origami with angular deficit and surplus}
We first employ the continuum mechanics approach to derive the kinematics from the intermediate state to the folded state for the three-fold origami, where the intermediate state has an angular deficit or surplus. Let $\tilde{\bfx}_i = \tilde{\delta}_i \tilde{\bft}_i$, where $\tilde{\bft}_i$ is the unit tangent of the crease, for $i=1,2,3,4$, in the intermediate state. Let the sector angle between $\tilde{\bft}_i$ and $\tilde{\bft}_{i+1}$  be $\tilde{\alpha}_{i}$, as shown in Fig.~\ref{fig:three_fold}(b). Then the tangents are given by
\beq
\tilde{\bft}_1 = \bfe_1,~ \tilde{\bft}_2= \bfR_{\bfe_3} (\tilde{\alpha}_1) \bfe_1, ~\tilde{\bft}_3 = \bfR_{\bfe_3}(\tilde{\alpha}_1 + \tilde{\alpha}_2) \bfe_1,~ \tilde{\bft}_4 = \bfR_{\bfe_3}(\tilde{\alpha}_1 + \tilde{\alpha}_2 + \tilde{\alpha}_3) \bfe_1. \label{eq:defomred_tangent1}
\eeq
The pattern with an angular deficit (surplus) satisfies $\tilde{\alpha}_1 + \tilde{\alpha}_2 + \tilde{\alpha}_3 < 2\pi$ ($\tilde{\alpha}_1 + \tilde{\alpha}_2 + \tilde{\alpha}_3 > 2\pi$).

By folding along the creases isometrically, the origami may form a continuous  configuration in 3D by matching $\tilde{\bft}_4$ and $\tilde{\bft}_1$, as illustrated in Fig.~\ref{fig:three_fold}. We employ the continuum mechanics approach to derive the folding angles to achieve the continuous configuration. The deformation gradients of the panels are $\bfI$, $\bfR_2(\gamma_2)$, and $\bfR_2(\gamma_2) \bfR_3(\gamma_3)$ respectively, where $\bfR_i (.) = \cos(.) (\tilde{\bft}_i \otimes \tilde{\bft}_i + \tilde{\bft}_i^{\perp} \otimes  \tilde{\bft}_i^{\perp}) + \sin(.) (-\tilde{\bft}_i \otimes \tilde{\bft}_i^{\perp} + \tilde{\bft}_i^{\perp} \otimes \tilde{\bft}_i) + \bfe_3 \otimes \bfe_3$ is a rotation tensor about $\tilde{\bft}_i$. Here $\tilde{\bft}_i^{\perp} = \bfR_{\bfe_3}(\pi/2) \tilde{\bft}_i$.  The continuity at $\tilde{\bft}_1$ (and $\tilde{\bft}_4$) requires the compatibility condition
	\beq
	\bfR_2(\gamma_2) \bfR_3(\gamma_3) \tilde{\bft}_4 = \tilde{\bft}_1 \label{eq:compatibility_3fold1}
	\eeq
	for both the deficit case and the surplus case. Satisfying the compatibility condition yields a folding deformation $\bfy(\tilde{\bfx})$ that maps the intermediate state to the folded state as
\beq
\bfy (\tilde{\bfx})  = \left\{ \begin{array}{ll} \tilde{\bfx}, 
	&  \tilde{\bfx} \cdot \bfe_3 = 0,\ \tilde{\bfx} \cdot \tilde{\bft}_2^\perp<0,\
	\tilde{\bfx} \cdot \tilde{\bft}_1^\perp \ge 0,  \\
	\bfR_2(\gamma_2) \tilde{\bfx}, &  \tilde{\bfx} \cdot \bfe_3 = 0,\ \tilde{\bfx} \cdot \tilde{\bft}_3^\perp <0,\
	\tilde{\bfx} \cdot \tilde{\bft}_2^\perp \ge 0, \\  
	\bfR_2(\gamma_2) \bfR_3(\gamma_3) \tilde{\bfx}, &  \tilde{\bfx} \cdot \bfe_3 = 0,\ \tilde{\bfx} \cdot \tilde{\bft}_4^\perp <0,\
	\tilde{\bfx} \cdot \tilde{\bft}_3^\perp \ge 0.
\end{array} \right. \label{y(x)}
\eeq
The remaining problem is to determine the folding angles $\gamma_2$ and $\gamma_3$, which are given by the following theorem. The theorem implies that there are two compatible deformed configurations at fixed folding angles $\gamma_2$ and $\gamma_3$.
	 \begin{theorem}
 Given $\tilde{\alpha}_1, \tilde{\alpha}_2, \tilde{\alpha}_3 \in (0, 2\pi)\setminus\{\pi\}$ as sector angles, the compatibility condition (\ref{eq:compatibility_3fold1}) holds if and only if 
\beqs
\gamma_2 &=& \delta \arccos \left( \frac{-\cos \tilde{\alpha}_3 + \cos\tilde{\alpha}_1 \cos\tilde{\alpha}_2}{\sin \tilde{\alpha}_1 \sin\tilde{\alpha}_2}\right), \nonumber \\
\gamma_3 &=& \delta \sign\left(\frac{\sin\tilde{\alpha}_1}{\sin\tilde{\alpha}_3}\right)\arccos \left( \frac{-\cos \tilde{\alpha}_1 + \cos\tilde{\alpha}_2 \cos\tilde{\alpha}_3}{\sin \tilde{\alpha}_2 \sin\tilde{\alpha}_3}\right), \label{eq:foldingangle_3fold}
\eeqs
for $\tilde{\alpha}_i$ in the admissible domain shown in Fig.~\ref{fig:ad_domain}(a), where $\delta\in{\pm 1}$ denotes the two branches.

	 \end{theorem}
 \begin{proof}
 	Projecting both sides of Eq.~(\ref{eq:compatibility_3fold1}) on $\tilde{\bft}_2$ yields
 	\beq
 	\bfR_2(\gamma_2) \bfR_3(\gamma_3) \tilde{\bft}_4 \cdot \tilde{\bft}_2 = \tilde{\bft}_1 \cdot \tilde{\bft}_2 \Leftrightarrow \cos \gamma_3 = \frac{-\cos \tilde{\alpha}_1 + \cos\tilde{\alpha}_2 \cos\tilde{\alpha}_3}{\sin \tilde{\alpha}_2 \sin\tilde{\alpha}_3}. \label{eq:cosgamma3_3}
 	\eeq
 	Similarly, dotting Eq.~(\ref{eq:compatibility_3fold1}) with $\bfR_2^{\mathrm{T}}\tilde{\bft}_3$ gives
 	\beq
 	\cos \gamma_2 = \frac{-\cos \tilde{\alpha}_3 + \cos\tilde{\alpha}_1 \cos\tilde{\alpha}_2}{\sin \tilde{\alpha}_1 \sin\tilde{\alpha}_2}. \label{eq:cosgamma2_3}
 	\eeq
To determine the sign of $\gamma_2$ and $\gamma_3$, we dot (\ref{eq:compatibility_3fold1}) with $\bfe_3$ and substitute $\cos\gamma_i$ to obtain
\beq
\sin\gamma_3 = \frac{\sin\tilde{\alpha}_1}{\sin\tilde{\alpha}_3} \sin\gamma_2.
\eeq
Thus the solution of $\gamma_2$ and $\gamma_3$ can be summarized as 
\beqs
\gamma_2 &=& \delta \arccos \left( \frac{-\cos \tilde{\alpha}_3 + \cos\tilde{\alpha}_1 \cos\tilde{\alpha}_2}{\sin \tilde{\alpha}_1 \sin\tilde{\alpha}_2}\right), \nonumber \\
\gamma_3 &=& \delta \sign\left(\frac{\sin\tilde{\alpha}_1}{\sin\tilde{\alpha}_3}\right)\arccos \left( \frac{-\cos \tilde{\alpha}_1 + \cos\tilde{\alpha}_2 \cos\tilde{\alpha}_3}{\sin \tilde{\alpha}_2 \sin\tilde{\alpha}_3}\right),
\eeqs
where $\delta \in \{\pm 1\}$ denotes the two branches.
    {Since $\cos\gamma_2, \cos\gamma_3 \in [-1,1]$, we have the admissible domain (Fig.~\ref{fig:ad_domain}(a)) for $\tilde{\alpha}_{1,2,3}$ such that the absolute values of Eqs.~(\ref{eq:cosgamma3_3}) and (\ref{eq:cosgamma2_3}) are less than or equal to 1. The symmetric shape of the domain arises from the symmetry of the three sector angles.}
 \end{proof}
 
\begin{figure}[h]
	\centering
	\includegraphics[width=\textwidth]{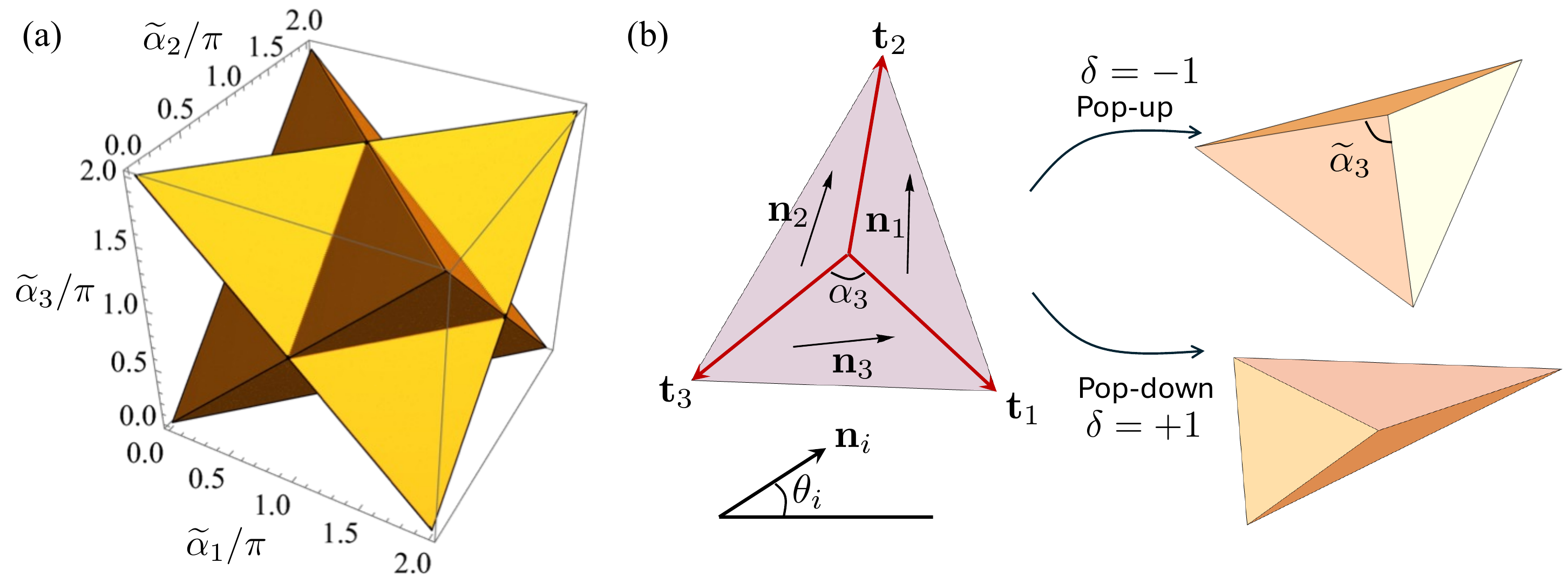}
	\caption{(a) The admissible domain of $\tilde{\alpha}_1$, $\tilde{\alpha}_2$ and $\tilde{\alpha}_3$, which ensures the folding angle functions (\ref{eq:cosgamma3_3}) and (\ref{eq:cosgamma2_3}) to have solutions. (b) Notation for a three-fold director field.}
	\label{fig:ad_domain}
\end{figure}

\subsection{Connection to director fields}
Now we consider the relationship between the reference director field and the final folded configuration. The final state can be determined directly from the initial director fields, eliminating the need to compute the intermediate state.

Let $\bfn_i =\bfR_{\bfe_3} (\theta_i) \bfe_1$, $i=1,2,3$ be the constant director on the $i-$th domain $\calS_i$ in a three-fold field. The three domains are separated by three interfaces with unit tangents $\bft_i, i=1,2,3$, where $\bft_i$ is the tangent of the interface between the domain $\calS_{i}$ and $\calS_{i-1}$ (here we define $\calS_0 = \calS_3$). Following the metric-compatible condition (\ref{eq:metric}) and the tangent solution (\ref{eq:tangent}), the unit tangent $\bft_i$ of the interface satisfies
\beq
\bft_i = \frac{\sigma_i \bfn_i + \tilde{\sigma}_i \bfn_{i-1}}{|\sigma_i \bfn_i + \tilde{\sigma}_i \bfn_{i-1}|}, \label{eq:tangent_3fold}
\eeq
where $\sigma_i, \tilde{\sigma}_i \in \{\pm 1 \}$. Equation~(\ref{eq:tangent_3fold}) implies that there are four solutions of $\bft_2$ for the prescribed $\bfn_1$ and $\bfn_2$. However, these tangents cannot be arbitrarily chosen from the four solutions when patterning a three-fold pattern$-$they need to obey the following topological condition 
\beq
\sign[(\bft_1 \times \bft_2) \cdot \bfe_3] = \sign[(\bft_2 \times \bft_3) \cdot \bfe_3] =\sign[(\bft_3 \times \bft_1) \cdot \bfe_3]  = +1 \label{eq:topolotial_condition1}
\eeq
to arrange $\bft_1, \bft_2, \bft_3$ counterclockwise and ensure that the sector angles are less than $\pi$. 

Denote the three vertices of the three-fold intersection by $\bfx_i =\delta_i \bft_i$ (no summation), where $\delta_i>0$ is the length of the crease with tangent $\bft_i$, $i=1,2,3$. Upon actuation, the panels remain flat and exhibit pure rotations about the crease. The three-fold intersection will pop up (or down) and form a dome-like structure if the sum of the actuated sector angles is less than $2 \pi$ (Fig.~\ref{fig:ad_domain}(b)). We derive the non-isometric deformation from its reference state to the final deformed state. First,
employing Eq.~(\ref{eq:sector_angle}), one can obtain the deformed sector angles as
\beq
\tilde{\alpha}_i = \arccos\left(\frac{\bfU_{\bfn_i} \bft_{i} \cdot \bfU_{\bfn_i}\bft_{i+1}}{|\bfU_{\bfn_i} \bft_{i}| |\bfU_{\bfn_i} \bft_{i+1}|}\right), \label{eq:deformed_sector_angle}
\eeq
where $i=1,2,3$, $\bft_4 := \bft_1$, and $\bft_i$ is given by (\ref{eq:tangent_3fold}). We then substitute the deformed sector angles $\tilde{\alpha}_i$ into the folding angle function (\ref{eq:foldingangle_3fold}), to obtain the folding angles $\gamma_i$ in terms of the direction of the director $\theta_i$, as stated in the following theorem.

\begin{figure}[ht]
    \centering
    \includegraphics[width=0.8\linewidth]{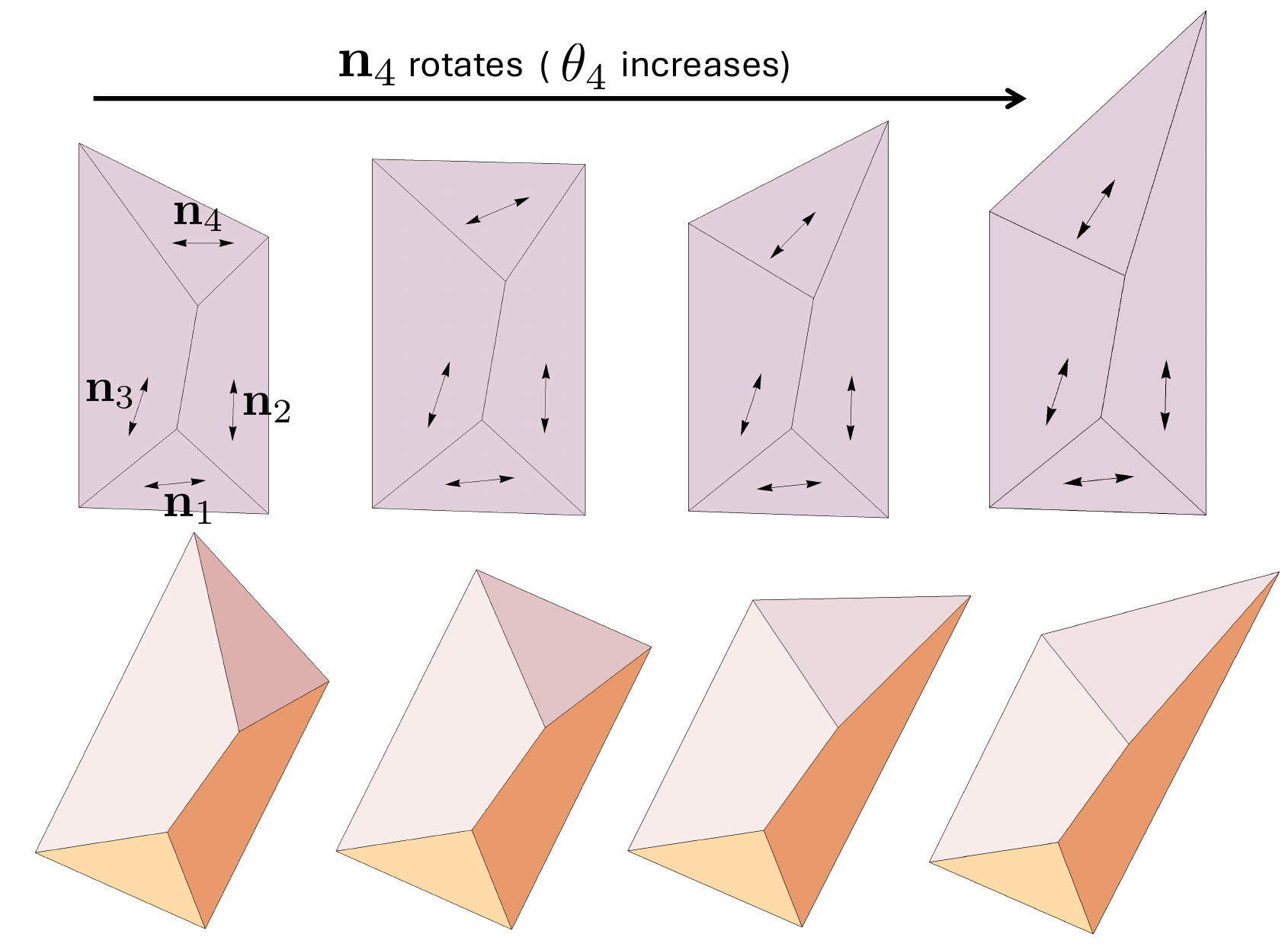}
    \caption{Director patterns consisting of two three-fold intersections with metric-compatible interfaces. The deformed configurations remain compatible when changing $\bfn_4$ but keeping $\bfn_1,\bfn_2,\bfn_3$ unchanged.}
    \label{fig:angle_change}
\end{figure}
\begin{theorem} \label{thm:3fold}
	Assume that the three-fold intersection of piecewise constant director fields shown in Fig.~\ref{fig:ad_domain}(b) satisfies the topological condition (\ref{eq:topolotial_condition1}) and the tangent condition (\ref{eq:tangent_3fold}). The directors are given by $\bfn_i = \bfR_{\bfe_3}(\theta_i) \bfe_1$, $i=1,2,3$. Then the folding angles $\gamma_2, \gamma_3$ satisfying the compatibility condition (\ref{eq:compatibility_3fold1}) are given by
\beq
\cos\gamma_i = (r^2)^{-\hat{\sigma}_1 \hat{\sigma}_{2} \hat{\sigma}_{3}} + \left(\hat{\sigma}_{i-1} \hat{\sigma}_{i+1}  - \hat{\sigma}_{i-1} \hat{\sigma}_{i+1}  (r^2)^{-\hat{\sigma}_1 \hat{\sigma}_2 \hat{\sigma}_3 } \right) \cos(\theta_i - \theta_{i-1}), \label{eq:thm_cosgamma}
\eeq
where $\hat{\sigma}_i = \sigma_i \tilde{\sigma}_i $.
	This leads to two solutions for folding angles $\gamma_2$ and $\gamma_3$:
	\beq
	\gamma_i =\pm \arccos\left[(r^2)^{-\hat{\sigma}_1 \hat{\sigma}_{2} \hat{\sigma}_{3}} + \left(\hat{\sigma}_{i-1} \hat{\sigma}_{i+1}  - \hat{\sigma}_{i-1} \hat{\sigma}_{i+1}  (r^2)^{-\hat{\sigma}_1 \hat{\sigma}_2 \hat{\sigma}_3 } \right) \cos(\theta_i - \theta_{i-1})\right], ~i=2,3,
	\eeq
where $\pm$ corresponds to the ``pop-up"  or the ``pop-down" deformation respectively, and $r=\lambda_\parallel / \lambda_\perp$.
\end{theorem}
\begin{proof}
We take $\cos \gamma_2$ as an example for the proof, and the other $\gamma_i$ can be obtained by cycling the indices. As shown in (\ref{eq:cosgamma2_3}), we obtain the folding angle $\gamma_2$ as a function of the intermediate sector angles
\beq
\cos\gamma_2 = \frac{-\cos\tilde{\alpha}_3 + \cos\tilde{\alpha}_1 \cos\tilde{\alpha}_2}{\sin\tilde{\alpha}_1 \sin \tilde{\alpha}_2}, \label{eq:cosgamma2}
\eeq
where the intermediate sector angles are given by (\ref{eq:deformed_sector_angle}). Specifically, we have
\beq
\cos\tilde{\alpha}_1 = \frac{\bfU_{\bfn_1}\bft_1 \cdot \bfU_{\bfn_1}\bft_2}{|\bfU_{\bfn_1}\bft_1| |\bfU_{\bfn_1}\bft_2|},~ \cos\tilde{\alpha}_2 = \frac{\bfU_{\bfn_2}\bft_2 \cdot \bfU_{\bfn_2}\bft_3}{|\bfU_{\bfn_2}\bft_2| |\bfU_{\bfn_2}\bft_3|},~
\cos\tilde{\alpha}_3 = \frac{\bfU_{\bfn_3}\bft_1 \cdot \bfU_{\bfn_3}\bft_3}{|\bfU_{\bfn_3}\bft_1| |\bfU_{\bfn_3}\bft_3|}.
\label{eq:cos}
\eeq
Notice that $|\bfU_{\bfn_1} \bft_2| = |\bfU_{\bfn_2} \bft_2|$, $|\bfU_{\bfn_2} \bft_3| = |\bfU_{\bfn_3} \bft_3|$, and $|\bfU_{\bfn_1} \bft_1| = |\bfU_{\bfn_3} \bft_1|$ by the metric compatibility. Substituting (\ref{eq:cos}) into (\ref{eq:cosgamma2}) yields
\beq
\cos\gamma_2 = \frac{(\bfU_{\bfn_1} \bft_1 \cdot \bfU_{\bfn_1} \bft_2) (\bfU_{\bfn_2} \bft_2 \cdot \bfU_{\bfn_2} \bft_3) - (\bfU_{\bfn_3} \bft_1 \cdot \bfU_{\bfn_3} \bft_3) |\bfU_{\bfn_2} \bft_2|^2}{\lambda_{\parallel}^2 \lambda_{\perp}^2(\bft_1 \times \bft_2)\cdot (\bft_2 \times \bft_3)}.
\label{eq:cosgamma22}
\eeq
The denominator of (\ref{eq:cosgamma22}) follows the equalities
\beqs
&&|\bfU_{\bfn_1} \bft_1| |\bfU_{\bfn_1} \bft_2| \sin\tilde{\alpha}_1 = \lambda_{\parallel} \lambda_{\perp} \sin\alpha_1 = \lambda_{\parallel} \lambda_{\perp} (\bft_1 \times \bft_2) \cdot \bfe_3,\\ \nonumber
&& |\bfU_{\bfn_2} \bft_2| |\bfU_{\bfn_2} \bft_3| \sin \tilde{\alpha}_2 = \lambda_{\parallel} \lambda_{\perp} \sin\alpha_2 = \lambda_{\parallel} \lambda_{\perp} (\bft_2 \times \bft_3) \cdot \bfe_3.
\eeqs
Substituting the expressions of $\bft_i$ and $\bfU_{\bfn_i}$ into (\ref{eq:cosgamma22}), in principal we have 64 different expressions for $\cos\gamma_2$, depending on the different choices of $\sigma_i$ and $\tilde{\sigma}_i$. However, we can reduce the number significantly. Notice that the set of choices $\{ \{\hat{\sigma}_1 \bft_1, \hat{\sigma}_2 \bft_2, \hat{\sigma}_3 \bft_3\}, \hat{\sigma}_i \in \{\pm 1\}, i=1,2,3 \}$ all give the same result for (\ref{eq:cosgamma22}). Therefore, we only need to check 8 cases. Let 
\beq
\bft_1 = \hat{\sigma}_1 \bfn_3 + \bfn_1,~ \bft_2 = \hat{\sigma}_2 \bfn_1 + \bfn_2,~ \bft_3 = \hat{\sigma}_3 \bfn_2 + \bfn_3,
\eeq
for $\hat{\sigma}_i \in \{\pm 1\}$. The direct computation leads to 8 solutions for $\cos\gamma_2$, classified as 
\begin{align}
\cos\gamma_2 = 
\begin{cases}
\frac{({\lambda_{\parallel}}^2-{\lambda_{\perp}}^2) \cos ({\theta_2}-{\theta_1})+{\lambda_{\perp}}^2}{{\lambda_{\parallel}}^2}, \quad \text{if}~\{\hat{\sigma}_1, \hat{\sigma}_2, \hat{\sigma}_3\} = \{+1,+1,+1\}~\text{or}~ \{-1,+1,-1\},  \\
\frac{({\lambda_{\parallel}}^2-{\lambda_{\perp}}^2) \cos ({\theta_2}-{\theta_1})+{\lambda_{\parallel}}^2}{{\lambda_{\perp}}^2}, \quad \text{if}~\{\hat{\sigma}_1, \hat{\sigma}_2, \hat{\sigma}_3\} = \{-1,+1,+1\}~\text{or}~ \{+1,+1,-1\},  \\
\frac{(-{\lambda_{\parallel}}^2+{\lambda_{\perp}}^2) \cos ({\theta_2}-{\theta_1})+{\lambda_{\parallel}}^2}{{\lambda_{\perp}}^2}, \quad \text{if}~\{\hat{\sigma}_1, \hat{\sigma}_2, \hat{\sigma}_3\} = \{+1,-1,+1\}~\text{or}~ \{-1,-1,-1\},  \\
\frac{(-{\lambda_{\parallel}}^2+{\lambda_{\perp}}^2) \cos ({\theta_2}-{\theta_1})+{\lambda_{\perp}}^2}{{\lambda_{\parallel}}^2}, \quad \text{if}~\{\hat{\sigma}_1, \hat{\sigma}_2, \hat{\sigma}_3\} = \{-1,-1,+1\}~\text{or}~ \{+1,-1,-1\}.
\end{cases} \label{eq:8_solutions}
\end{align} 
Summarizing all the above solutions, we have a unified expression 
\beq
\cos\gamma_2 = (r^2)^{-\hat{\sigma}_1 \hat{\sigma}_2 \hat{\sigma}_3}+ \left(\hat{\sigma}_1 \hat{\sigma}_3  - \hat{\sigma}_1 \hat{\sigma}_3  (r^2)^{-\hat{\sigma}_1 \hat{\sigma}_2 \hat{\sigma}_3 } \right) \cos(\theta_2 - \theta_1), \label{eq:proof_3fold}
\eeq
where $r= \lambda_{\parallel} / \lambda_{\perp}$. For all the 64 cases, one can replace $\hat{\sigma}_i$ by $\sigma_i \tilde{\sigma}_i$, which leads to  Eq.~(\ref{eq:thm_cosgamma}) for $i=2$. Cycling $i=1,2,3$ in (\ref{eq:proof_3fold}) yields other folding angles and completes the proof.
\end{proof}

An interesting observation from Theorem~\ref{thm:3fold} is that, fixing $\sigma_i$ and $\tilde{\sigma}_i$, the folding angle $\gamma_i$ depends only on the directors on its adjacent panels, i.e., $\bfn_i$ and $\bfn_{i-1}$ with the directional angles $\theta_i$ and $\theta_{i-1}$. This fact provides more freedom to design connected three-fold intersections.
For example, in Fig.~\ref{fig:angle_change}, the pattern consists of two three-fold intersections.
We may keep $\bfn_2$ and $\bfn_3$ unchanged and vary the direction of $\bfn_4$. {As a result of 
Theorem~\ref{thm:3fold}, the folding angle along the shared crease in the actuated configuration is independent of $\bfn_4$. The folding angle along the shared crease will remain unchanged. Thus, the actuated configuration containing the two three-fold nodes remains compatible when changing $\bfn_4$.} We will use the basic pattern in Fig.~\ref{fig:angle_change} to design more complex patterns in Section 5.

\section{Four-fold intersection}
{
In this section, we study the kinematics of a general non-Euclidean four-fold origami generated by piecewise constant director fields. 
The kinematics of the four-fold origami (or more general n-fold origami) has been extensively studied (Euclidean or non-Euclidean) by Huffman \cite{huffman1976curvature}, Santangelo \cite{evans2015lattice}, van Hecke \cite{waitukaitis2015origami, waitukaitis2020non}, Tachi \cite{tachi2009generalization}, Hull \cite{hull1994mathematics}, Lang and Howell \cite{lang2018rigidly}, Guest \cite{he2020rigid}, Izmestiev \cite{izmestiev2017classification}, and more, using various methods (e.g. spherical geometry, structural kinematics, mechanisms, etc). However, the continuum-mechanics framework developed in our previous work \cite{feng2020designs}—which addressed the Euclidean, flat-foldable case—remains worth exploring in the non-Euclidean regime. In particular, we will derive explicit expressions for the folding angle functions and corresponding deformations within the Lagrangian framework, expressed in terms of a reference folding angle defined over the admissible domain. This formulation provides a homotopy that accurately characterizes the evolution of shapes, which is convenient to use in an algorithm.}

Again, we deform the intermediate state isometrically to achieve a continuous folded state. The method we use to compute the folding angle functions is partially adapted from our previous work \cite{zou2024kinematics}. Here, we use a new symmetry approach to make the folding angle functions rigorous and more convenient to use for computation. We also provide a more physically relevant understanding of the ``branches" in Fig.~\ref{fig:reference}.

\subsection{Reference domain}
Suppose the intermediate state is given in Figure \ref{fig:fourfold}(a), where $0 <  \tilde{\alpha}_1, \tilde{\alpha}_2, \tilde{\alpha}_3, \tilde{\alpha}_4 < \pi$
are the sector angles.
The tangents of the creases are given by
\beq
\tilde{\bft}_1 = \bfe_1,~ \tilde{\bft}_2 = \bfR_{\bfe_3}(\tilde{\alpha}_1) \bfe_1,~ \tilde{\bft}_3 = \bfR_{\bfe_3}(\tilde{\alpha}_1+ \tilde{\alpha}_2) \bfe_1, ~\tilde{\bft}_4 = \bfR_{\bfe_3}(\tilde{\alpha}_1 + \tilde{\alpha}_2 + \tilde{\alpha}_3) \bfe_1,~ \tilde{\bft}_5 = \bfR_{\bfe_3}(\tilde{\alpha}_1+ \tilde{\alpha}_2 + \tilde{\alpha}_3 + \tilde{\alpha}_4) \bfe_1.
\label{tees}
\eeq
In general, the sum of the four sector angles $\sum_{i=1}^{4}\alpha_i$ is not $2\pi$, resulting in nonzero Gaussian curvature at the tip in the actuated configuration. Also, we have $\tilde{\bft}_5 \neq \tilde{\bft}_1$ and a deficit (or surplus) may exist in the reference domain. Alternatively, a second version of the intermediate domain can be obtained with a cut between $\tilde{\bft}_4$ and $ \tilde{\bft}_{55}$ where $\tilde{\bft}_{55}=\cos\tilde{\alpha}_4 \bfe_1 - \sin\tilde{\alpha}_4 \bfe_2$, as shown in Figure \ref{fig:fourfold}(a). 
Combining the equivalent two domains is convenient for computing the folding angles in the next section.

\begin{figure}[ht]
	\center
	\includegraphics[width=\textwidth]{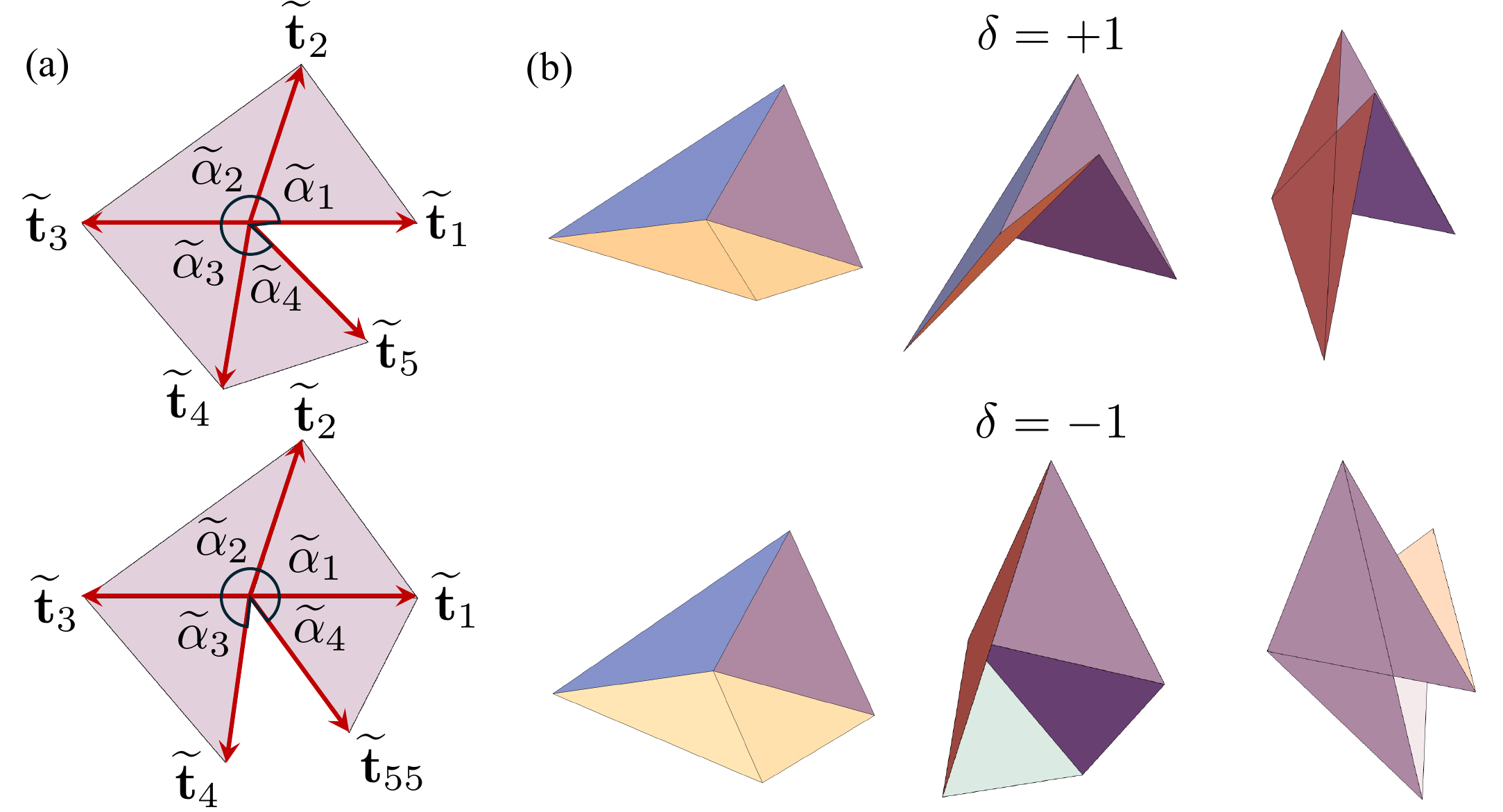}
	\caption{(a) Two equivalent intermediate configurations. (b) Two folding branches denoted by $\delta=+1$ and $\delta=-1$.} \label{fig:fourfold}
\end{figure}

Let $\bfR_i(\theta) \in \text{SO(3)} $ be a rotation with axis $\tilde{\bft}_i$. Similarly to Eq.~(\ref{y(x)}), the isometric deformation from the intermediate state (Fig.~\ref{fig:fourfold}(a) top) to the folded state (Fig.~\ref{fig:fourfold}(b)) is described by
\beq
\bfy (\tilde{\bfx})  = \left\{ \begin{array}{ll} \tilde{\bfx}, 
	&  \tilde{\bfx} \cdot \bfe_3 = 0,\ \tilde{\bfx} \cdot \tilde{\bft}_2^\perp<0,\
	\tilde{\bfx} \cdot \tilde{\bft}_1^\perp \ge 0,  \\
	\bfR_2(\gamma_2) \tilde{\bfx}, &  \tilde{\bfx} \cdot \bfe_3 = 0,\ \tilde{\bfx} \cdot \tilde{\bft}_3^\perp <0,\
	\tilde{\bfx} \cdot \tilde{\bft}_2^\perp \ge 0, \\  
	\bfR_2(\gamma_2) \bfR_3(\gamma_3) \tilde{\bfx}, &  \tilde{\bfx} \cdot \bfe_3 = 0,\ \tilde{\bfx} \cdot \tilde{\bft}_4^\perp <0,\
	\tilde{\bfx} \cdot \tilde{\bft}_3^\perp \ge 0, \\
    \bfR_2(\gamma_2) \bfR_3(\gamma_3) \bfR_4(\gamma_4) \tilde{\bfx}, &  \tilde{\bfx} \cdot \bfe_3 = 0,\ \tilde{\bfx} \cdot \tilde{\bft}_5^\perp <0,\
	\tilde{\bfx} \cdot \tilde{\bft}_4^\perp \ge 0.
\end{array} \right. \label{eq:y(x)fourfold}
\eeq
The deformation is continuous except at the $\tilde{\bft}_1 \tilde{\bft}_5$ boundary. Thus, a compatibility condition is needed to ensure continuity, which leads to the solutions of the folding angles in the next section. Using the second version of the intermediate domain (Fig.~\ref{fig:fourfold}(a) bottom), the deformation in the fourth sector domain $ \{\tilde{\bfx} \cdot \bfe_3 = 0,\ \tilde{\bfx} \cdot \tilde{\bft}^\perp_1<0,\
\tilde{\bfx} \cdot \tilde{\bft}^\perp_{55}\ge 0\}$ can be written as
\beq
\bfy(\tilde{\bfx})= \bfR_1(-\gamma_1) \tilde{\bfx}, \label{eq:gamma1_con}
\eeq
which includes the folding angle $\gamma_1$.

\subsection{Compatibility condition and folding angle functions}
To ensure the continuity at the $\tilde{\bft}_1 \tilde{\bft}_5$ boundary, the deformation gradients need to satisfy the compatibility condition:
\beq
\bfR_2(\gamma_2) \bfR_3(\gamma_3) \bfR_4(\gamma_4)\tilde{\bft}_5 = \tilde{\bft}_1. \label{eq:compatibility_fourfold}
\eeq
The folding angles are restricted in the interval $-\pi \le \gamma_1, \gamma_2, \gamma_3, \gamma_4 \le \pi$.  In the rigidly and flat-foldable case, the four folding angles reach $\pm \pi$ simultaneously, with a reference free folding angle ranging from $-\pi$ to $\pi$. In the non-Euclidean case, some folding angles will reach $\pm \pi$ first, making the domain of the reference folding angle not a full interval of $[-\pi,\pi]$.

\subsubsection{Calculations of $\gamma_2$, $\gamma_3$ and $\gamma_4$}
To solve the compatibility equation, we first multiply (\ref{eq:compatibility_fourfold}) by $\bfR_2^T(\gamma_2)$ to have 
\beq
\bfR_3(\gamma_3) \bfR_4(\gamma_4)\tilde{\bft}_5 = \bfR_2^T(\gamma_2) \tilde{\bft}_1, \label{eq:gamma3_1}
\eeq
and then dot (\ref{eq:gamma3_1}) with $\tilde{\bft}_3$ to have a scalar equation:
\beq
\bfR_4(\gamma_4) \tilde{\bft}_5 \cdot \tilde{\bft}_3 - \bfR_2^T(\gamma_2) \tilde{\bft}_1 \cdot \tilde{\bft}_3 = 0
\eeq
(recalling $\bfR_3(\gamma_3) \tilde{\bft}_3 = \tilde{\bft}_3$). This explicitly gives 
\beqs
\cos(\gamma_4)&=&\frac{\cos \tilde{\alpha}_3 \cos\tilde{\alpha}_4 - \cos\tilde{\alpha}_1 \cos\tilde{\alpha}_2 + \sin\tilde{\alpha}_1 \sin\tilde{\alpha}_2 \cos{\gamma}_2}{\sin\tilde{\alpha}_3 \sin\tilde{\alpha}_4} \nonumber \\
&:=&F_4(\gamma_2; \tilde{\boldsymbol{\alpha}}), \label{eq:f4}
\eeqs
where we define $\tilde{\boldsymbol{\alpha}}:=(\tilde{\alpha}_1,\tilde{\alpha}_2,\tilde{\alpha}_3,\tilde{\alpha}_4)$. It is convenient to choose $\gamma_2$ as the reference folding angle, say, $\gamma_2 = \theta$. Then,
$\gamma_4$ as a function of $\theta$ is given in the condensed form
\beq \label{eq:gamma4}
\gamma_4 =  \bar{\gamma}_4^{\delta}( \theta) := \delta \arccos F_4(\theta; \tilde{\boldsymbol{\alpha}}).
\eeq
Here $\delta \in\{\pm\}$ denotes the two branches of solution as shown in Fig.~\ref{fig:fourfold}(b). Since the function $\bar{\gamma}_4^\delta(\theta)$ is even and monotonic in $[-\pi,0]$ and $[0,\pi]$, we may characterize the admissible domain of $\theta$ by $\calI^-\cup\calI^+ := [-\theta_{\max},-\theta_{\min}] \cup [\theta_{\min}, \theta_{\max}]$ in the following cases:
\begin{enumerate}
	\item $F_4(0;\tilde{\boldsymbol{\alpha}}) \in [-1, 1]$. For this case, $\theta_{\min} = 0$, and $\theta_{\max}$ is given by
	\begin{align}
	\theta_{\max}= \begin{cases}
	\pi,  &\text{if}~F_4(\pi;\tilde{\boldsymbol{\alpha}}) \geq -1 \\
	\arccos \left[\frac{\cos\tilde{\alpha}_1 \cos\tilde{\alpha}_2 - \cos(\tilde{\alpha}_3-\tilde{\alpha}_4)}{\sin\tilde{\alpha}_1\sin\tilde{\alpha}_2}\right], &\text{if}~F_4(\pi; \tilde{\boldsymbol{\alpha}})<-1
	\end{cases}.
	\end{align}
	
	\item $F_4(0; \tilde{\boldsymbol{\alpha}}) >1$ and $F_4(\pi;\tilde{\boldsymbol{\alpha}}) \leq 1$. For this case, by solving $F_4(\theta_{\min}; \tilde{\boldsymbol{\alpha}}) = 1$ and $F_4(\theta_{\max}; \tilde{\boldsymbol{\alpha}}) = -1$, we have 
	\begin{align}
	&\theta_{\min}= \arccos \left[\frac{\cos\tilde{\alpha}_1\cos\tilde{\alpha}_2 - \cos(\tilde{\alpha}_3+\tilde{\alpha}_4)}{\sin\tilde{\alpha}_1\sin\tilde{\alpha}_2}\right], \\
	&\theta_{\max}= \begin{cases}
	\pi, \quad &\text{if}~F_4(\pi;\tilde{\boldsymbol{\alpha}}) \geq -1 \\
	\arccos \left[\frac{\cos\tilde{\alpha}_1\cos-\tilde{\alpha}_2 - \cos(\tilde{\alpha}_3-\tilde{\alpha}_4)}{\sin\tilde{\alpha}_1\sin\tilde{\alpha}_2}\right], \quad  & \text{if}~F_4(\pi; \tilde{\boldsymbol{\alpha}})<-1
	\end{cases}.
	\end{align}
	\item Otherwise, if $F_4(0; \tilde{\boldsymbol{\alpha}})$ and $F_4(\pi; \tilde{\boldsymbol{\alpha}})$ are beyond the above range, the admissible set of $\theta$ is empty.
\end{enumerate}

Next, we dot (\ref{eq:compatibility_fourfold}) with $\bfR_2^T(\gamma_2)$ and then project onto the plane perpendicular to $\tilde{\bft}_3$ to have
\beq
\bfR_3(\gamma_3)((\bfI - \tilde{\bft}_3 \otimes \tilde{\bft}_3)\bfR_4(\bar{\gamma}_4^{\delta}(\gamma_2))\tilde{\bft}_5)=  (\bfI - \tilde{\bft}_3 \otimes \tilde{\bft}_3) \bfR_2^T(\gamma_2) \tilde{\bft}_1 .   \label{eq:projection}
\eeq
The norms of the projected vectors are the same:
\beq
|(\bfI - \tilde{\bft}_3 \otimes \tilde{\bft}_3)(\bfR_3(\gamma_3) \bfR_4(\bar{\gamma}_4^{\sigma}(\gamma_2))\tilde{\bft}_5| = |(\bfI - \tilde{\bft}_3 \otimes \tilde{\bft}_3)\bfR_2^T(\gamma_2) \tilde{\bft}_1 |,
\eeq
which is trivially true.
Thus, to solve the projected equation (\ref{eq:projection}), we  consider the non-degenerate case under $|(\bfI - \tilde{\bft}_3 \otimes \tilde{\bft}_3)\bfR_2^T(\gamma_2) \tilde{\bft}_1 | \neq 0$. 
For (\ref{eq:projection}) to hold, $\gamma_3 = \bar{\gamma}_3^{\delta}(\gamma_2)$ must satisfy 
\beqs
\cos(\bar{\gamma}_3^{\delta}(\gamma_2))&=&\frac{(\bfI - \tilde{\bft}_3 \otimes \tilde{\bft}_3)\bfR_4(\bar{\gamma}_4^{\delta}(\gamma_2))\tilde{\bft}_5 \cdot (\bfI - \tilde{\bft}_3 \otimes \tilde{\bft}_3)\bfR_2^T(\gamma_2)\tilde{\bft}_1}{|(\bfI - \tilde{\bft}_3 \otimes \tilde{\bft}_3)\bfR_4(\bar{\gamma}_4^{\delta}(\gamma_2))\tilde{\bft}_5 | | (\bfI - \tilde{\bft}_3 \otimes \tilde{\bft}_3)\bfR_2^T(\gamma_2)\tilde{\bft}_1|} \nonumber\\ &:=& F_{3c}^{\delta}(\gamma_2; \tilde{\boldsymbol{\alpha}}),  \label{eq:cosxi}\\
\sin(\bar{\gamma}_3^{\delta}(\gamma_2)) &=&\text{sign}\left[\tilde{\bft}_3 \cdot (\bfI - \tilde{\bft}_3 \otimes \tilde{\bft}_3)\bfR_4(\bar{\gamma}_4^{\delta}(\gamma_2))\tilde{\bft}_5 \times (\bfI - \tilde{\bft}_3 \otimes \tilde{\bft}_3)\bfR_2^T(\gamma_2)\tilde{\bft}_1 \right]\sqrt{1-\cos^2(\bar{\gamma}_3^{\delta}(\gamma_2))}  \nonumber \\
&:=&F_{3s}^{\delta}(\gamma_2; \tilde{\boldsymbol{\alpha}}) \label{eq:sinxi}.
\eeqs
Then the solution of $\gamma_3$ is given by
\beq \label{eq:gamma3}
\gamma_3 =\bar{\gamma}_3^{\delta}(\gamma_2)=\text{sign}[F_{3s}^{\delta}(\gamma_2; \tilde{\boldsymbol{\alpha}})] \arccos[F_{3c}^{\delta}(\gamma_2; \tilde{\boldsymbol{\alpha}})]. 
\eeq

\begin{figure}[!ht]
	\center
	\includegraphics[width=\textwidth]{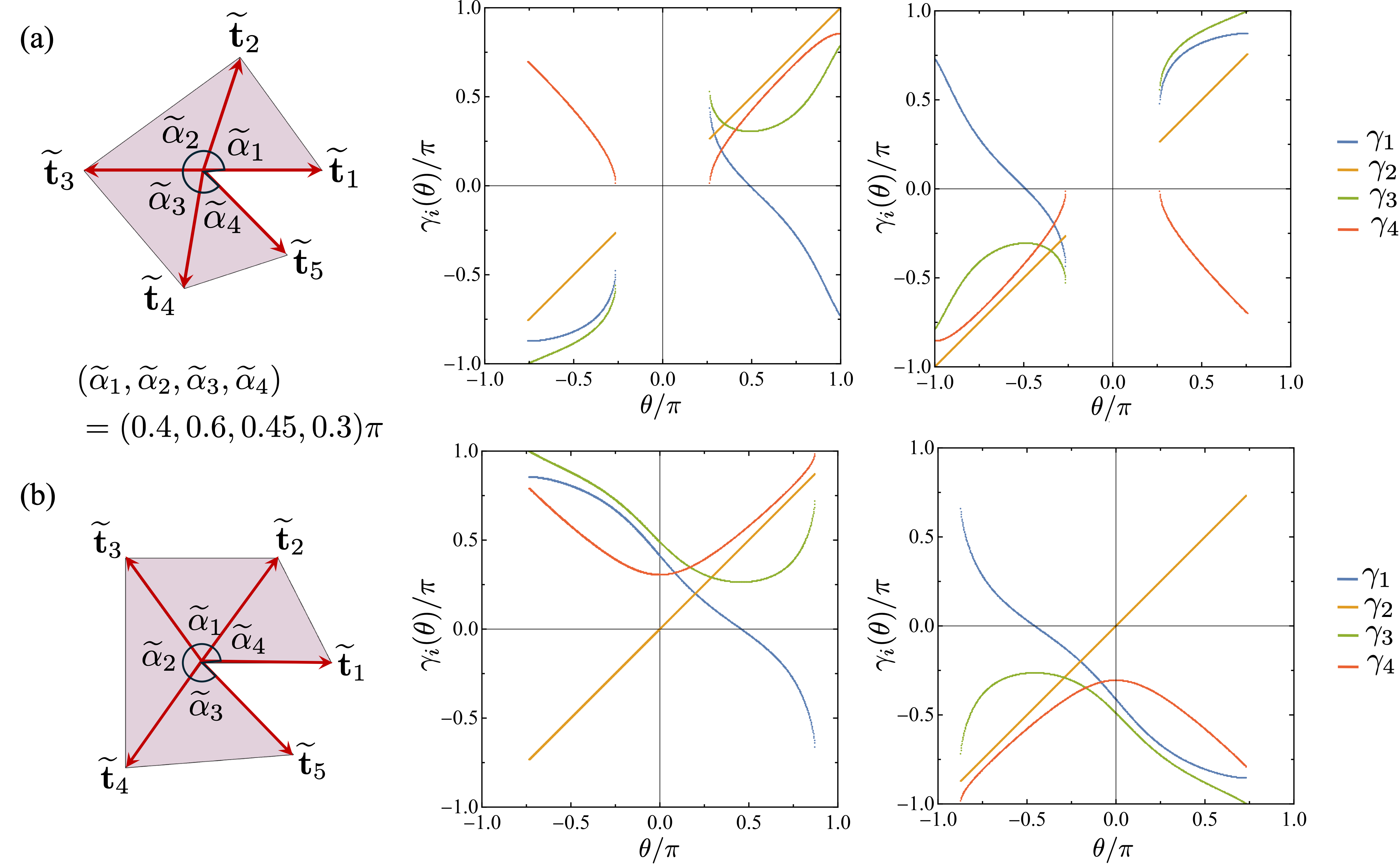}
	\caption{Two examples with the same set of sector angles. (a) The folding angle functions show discontinuities by choosing $\gamma_2$ as the reference folding angle. (b) The folding angle functions are continuous by cycling sector angles and then choosing $\gamma_2$ (which is $\gamma_1$ in (a)) as the reference folding angle.} \label{fig:reference}
\end{figure}
\subsubsection{Calculation of $\gamma_1$}
To solve $\gamma_1$, we use the second intermediate domain in Fig.~\ref{fig:fourfold}(a) such that there exists
$\gamma_1=\bar{\gamma}_1^\delta(\theta)$ making the folded configuration continuous across $\tilde{\bft}_4$ and $\tilde{\bft}_{55}$. A convenient way to obtain $\gamma_1$ is to use the cyclic symmetry of $\tilde{\alpha}_1, \tilde{\alpha}_2, \tilde{\alpha}_3$ and $\tilde{\alpha}_4$. Notice that $\gamma_4=\bar{\gamma}_4^\delta(\gamma_2)$ shows the ability to solve the folding angle ($\gamma_4$) in terms of its opposite folding angle ($\gamma_2$). Thus we may solve $\gamma_1$ in terms of its opposite folding angle $\gamma_3$, by simply cycling $\tilde{\alpha}_1, \tilde{\alpha}_2, \tilde{\alpha}_3, \tilde{\alpha}_4$  and then using the the same function form $F_4$ in (\ref{eq:f4}) as follows:
\beq
\gamma_1= \hat{\delta} \arccos F_4[ (\bar{\gamma_3}^\delta(\gamma_2; \tilde{\boldsymbol{\alpha}}); \hat{\boldsymbol{\alpha}}]:= \bar{\gamma}_1^{\hat{\delta}} (\gamma_2,\delta), \label{eq:gamma1}
\eeq
where $\hat{\boldsymbol{\alpha}} = (\tilde{\alpha}_2, \tilde{\alpha}_3, \tilde{\alpha}_4, \tilde{\alpha}_1)$. Substituting $\gamma_3$ as a function of $\gamma_2$, we can still use $\gamma_2$ as the reference folding angle in (\ref{eq:gamma1}). To determine the folding branch $\hat{\delta}$, 
 we perform the following continuity check for $\tilde{\bft}_4$ and $\tilde{\bft}_{55}$:
\begin{align}
\hat{\delta}(\theta) = \begin{cases}
+1,~ \text{if}~ \bfR_1(-\arccos F_4 (\bar{\gamma_3}^\delta(\theta; \tilde{\boldsymbol{\alpha}}); \hat{\boldsymbol{\alpha}})) \tilde{\bft}_{55} = \bfR_2 (\theta) \bfR_3(\bar{\gamma}_3^\delta(\theta)) \tilde{\bft}_4\\
-1,~ \text{if}~ \bfR_1(\arccos F_4 (\bar{\gamma_3}^\delta(\theta; \tilde{\boldsymbol{\alpha}}); \hat{\boldsymbol{\alpha}})) \tilde{\bft}_{55} = \bfR_2 (\theta) \bfR_3(\bar{\gamma}_3^\delta(\theta)) \tilde{\bft}_4
\end{cases}.
\end{align}
Then the deformation of the fourth domain is accordingly $\bfy(\tilde{\bfx})= \bfR_1(-\gamma_1) \tilde{\bfx}$.

\subsubsection{Admissible domain to avoid self-intersection}
In our previous analysis, we choose $\gamma_2$ as the reference folding angle, and the folding angle functions ensure the continuity of the deformation $\bfy(\tilde{\bfx})$. However, this analysis allows unphysical self-intersections, and one panel may be folded to pass through its adjacent panels. Thus, we need to determine the admissible domain of the reference folding angle. The domain $\calI^+ \cup \calI^-$ is obtained only from the existence of $\gamma_4$. Therefore, the admissible domain is a subset of $\calI^+ \cup \calI^-$ such that
 no self-intersection occurs at any crease.

Notice that the folding angle jumps between $-\pi$ and $\pi$ when the panel intersects with its adjacent panel. Thus, we perform the following check
\beq
|\gamma_i(\bar{\theta}^+) - \gamma_i(\bar{\theta}^-)| = 2 \pi
\eeq
to determine the reference folding angle $\bar{\theta}$ at which the folding angles $\gamma_i$ reaches $\pm \pi$. 
Let ${\cal J}^+ \subset \calI^+ = [\theta_{\min}, \theta_{\max}]$  be the set of points at which one of the folding angles jumps, say,
\beq
{\cal J}^+ :=\{\bar{\theta}\in\calI^+: |\gamma_i(\bar{\theta}^+) - \gamma_i(\bar{\theta}^-)| = 2 \pi, ~i=1,2,3 \text{ or }4\}.
\eeq
Define
\begin{align}
\bar{\theta}_{\max}:=\begin{cases}
\inf {\cal J}^+,\quad &\text{if}~{\cal J}^+ \neq \emptyset \\
\theta_{\max},\quad &\text{if}~{\cal J}^+ = \emptyset
\end{cases}
\end{align}
as the upper bound of the physically relevant reference folding angle (i.e., no self-intersection) on $\calI^+$. Similarly, the lower bound $\bar{\theta}_{\min}$ on $\calI^-=[-\theta_{\max},-\theta_{\min}]$ is defined as 
\begin{align}
\bar{\theta}_{\min}:=\begin{cases}
\sup {\cal J}^-,\quad &\text{if}~{\cal J}^- \neq \emptyset \\
-\theta_{\max},\quad &\text{if}~{\cal J}^- = \emptyset
\end{cases}
\end{align}
where ${\cal J}^-$ is given by
\beq
{\cal J}^- :=\{\bar{\theta} \in \calI^-: |\gamma_i(\bar{\theta}^+) - \gamma_i(\bar{\theta}^-)| = 2 \pi, ~i=1,2,3~\text{or}~4\}.
\eeq
Then the admissible domain of the reference folding angle is $\bar{\calI}^+\cup \bar{\calI}^-$, where 
\beq
\bar{\calI}^+ = [\theta_{\min}, \bar{\theta}_{\max}],\quad \bar{\calI}^- = [\bar{\theta}_{\min}, - \theta_{\min}].
\eeq
In reality, the domain can be obtained numerically by simply performing the above check, resulting in the folding angle curves in Fig.~\ref{fig:reference} that do not exceed $\pm \pi$.

In summary, the folding angle functions that make the folded configuration continuous with no self-intersections are given by  
\beq
\gamma_1=\bar{\gamma}_{1}^{\tilde{\delta}(\theta)}(\theta,\delta),~\gamma_2=\theta,~\gamma_3 =\bar{\gamma}_3^{\delta}(\theta),~\gamma_4 = \bar{\gamma}_4^\delta(\theta),~ \delta\in\{\pm 1\},~ \theta \in \bar{\calI}^-\cup \bar{\calI}^+. \label{eq:gamma1234}
\eeq
The Mathematica code for simulating the deformation of a general non-Euclidean four-fold origami is available upon request.

{\bf Remark 1: Symmetry relation.} By checking the folding angle functions, we claim that if $(\gamma_1, \gamma_2, \gamma_3, \gamma_4)$ is on the branch $\delta=+$, then $(-\gamma_1, -\gamma_2, -\gamma_3, -\gamma_4)$ is on the other branch $\delta=-$, as shown in Fig.~\ref{fig:reference}. These two configurations are mirror related.

{\bf Remark 2: Choices of different reference folding angles.} It is noticed that, for the same crease pattern, choosing the appropriate reference folding angle may result in continuous folding angle functions (Fig.~\ref{fig:reference}(b)) rather than discontinuous ones (Fig.~\ref{fig:reference}(a)). Using the notations in Fig.~\ref{fig:reference}(a), it is numerically verified that the folding angle functions are continuous in terms of $\gamma_2$ if the sector angles satisfy
\begin{align}
    \begin{cases}
        \tilde{\alpha}_1 +  \tilde{\alpha}_2 \leq  \tilde{\alpha}_3 +  \tilde{\alpha}_4 \\
         \tilde{\alpha}_1 + \tilde{\alpha}_2 + \tilde{\alpha}_3 + \tilde{\alpha}_4  \leq 2\pi
    \end{cases} \quad \text{or} \quad 
    \begin{cases}
        \tilde{\alpha}_1 +  \tilde{\alpha}_2 \geq  \tilde{\alpha}_3 +  \tilde{\alpha}_4 \\
         \tilde{\alpha}_1 + \tilde{\alpha}_2 + \tilde{\alpha}_3 + \tilde{\alpha}_4  \geq 2\pi
    \end{cases} .
\end{align}

{
\subsection{Connection to director fields}
As discussed in Section \ref{sec:nfold}, for the four-fold intersection, we only need to consider the degenerate case satisfying $\alpha_1+\alpha_3=\pi$ (or equivalently $\alpha_2+\alpha_4=\pi$). The example in Fig.~\ref{fig:fourfold_pattern} also satisfies this condition. Assuming $\bfn_1= \bfR_{\bfe_3}(\theta) \bfe_1$, we compute the director at each panel, and then
 calculate the deformed sector angles by $\tilde{\alpha}_i = \arccos[(\bfU_{\bfn_i}\bft_i \cdot \bfU_{\bfn_i}\bft_{i+1})/(|\bfU_{\bfn_i}\bft_i| |\bfU_{\bfn_i}\bft_{i+1}|)]$. Substituting these deformed sector angles into the folding angle functions yields the kinematics of the four-fold origami directly from the director pattern. 

 Suppose the reference sector angles are given by $(\alpha_1,\alpha_2,\alpha_3,\alpha_4) = (\alpha_1,\alpha_2,\pi-\alpha_1,\pi-\alpha_2)$, $\bfn_1= \bfR_{\bfe_3} (\theta) \bfe_1$ and $\bft_1=\bfe_1$. By using $\bfn_{i+1} = \bfn_i - 2 (\bfn_i \cdot \bft_{i+1}) \bft_{i+1}$, we compute the directors $\bfn_i = \bfR_{\bfe_3}(\theta_i)\bfe_1$ as 
 \beq
\theta_1=\theta,~\theta_2= 2\alpha_1 + 2\alpha_2-\theta+\pi,~ \theta_3 = \theta - 2\beta_1,~ \theta_4 = 2\beta_1 - \theta + \pi.
 \eeq
We then substitute the directors and the crease tangents to obtain the deformed sector angles $\tilde{\alpha}_i$. A further calculation of Eq.~(\ref{eq:f4}) yields (numerically validated)
\beqs
\cos(\gamma_4) &=& \frac{
\cos(\gamma_2)\!\left[1 + r^2 + (-1 + r^2)\cos(4\alpha_1 - 2\theta)\right]
 - 2(-1 + r^2)\sin(\alpha_1 - \alpha_2)\sin(3\alpha_1 + \alpha_2 - 2\theta)
}{
1 + r^2 + (-1 + r^2)\cos\!\big[2(\alpha_1 + \alpha_2 - \theta)\big]
}\nonumber\\
&:=& \tilde{F}_4(\gamma_2; \boldsymbol{\alpha},\theta) \label{eq:cosgamma4ref}
\eeqs
where $\boldsymbol{\alpha}=(\alpha_1,\alpha_2,\alpha_3,\alpha_4)=(\alpha_1,\alpha_2,\pi-\alpha_1,\pi-\alpha_2)$ and $r=\lambda_{\parallel}/\lambda_{\perp}$. 
Equation~(\ref{eq:cosgamma4ref}) gives the relationship between the opposite folding angles $\gamma_4$ and $\gamma_2$ in terms of the directors and creases (determined by $\theta$ and $\alpha_i$). Unfortunately, there is no such neat formula for the adjacent folding angles. But we can always obtain the deformed sector angles $\tilde{\alpha}_i$ in terms of $\theta$ and $\alpha_i$, and then compute the kinematics, because all the formulae including the folding angle functions (\ref{eq:sinxi}) and (\ref{eq:f4}) are explicit. The explicit forms of these functions provide the foundation for efficient numerical algorithms used to pattern the complex director fields in the next section.

}

\section{Patterning the three-fold and four-fold intersections}\label{sec:pattern}
It is important to pattern the three-fold and four-fold intersections into active metamaterials, relying on the kinematics of single nodes. 
The active metamaterials are designed to be flat before actuation, and will actuate into an effective plate/shell with concentrated Gaussian curvature at nodes after actuation. Here, we provide two systematic designs: quadrilateral director patterns consisting of four-fold nodes and director patterns consisting of both three-fold and four-fold nodes.

\subsection{Quadrilateral director patterns}
\begin{figure}[!ht]
    \centering
    \includegraphics[width=\linewidth]{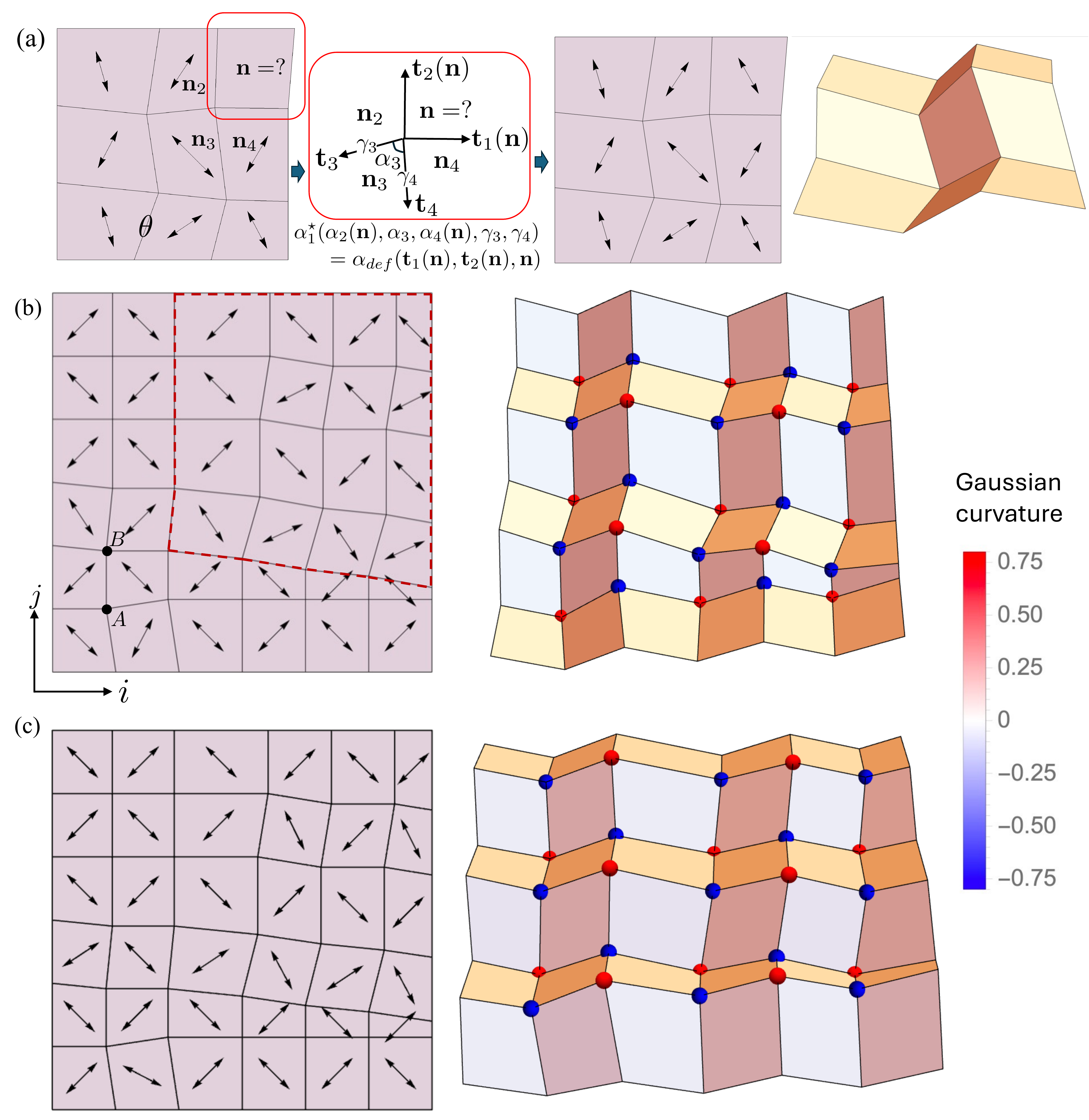}
    \caption{(a) The director on the $(3,3)$ panel (in the red frame) is determined by the other directors and the folding angles $\gamma_3,\gamma_4$. (b) A marching algorithm: the director pattern within the dashed red frame is determined by the boundary directors and the reference folding angle, following (a). (c) The orthogonal dual of (b). The patterns are actuated by $\lambda_{\parallel}=0.9, \lambda_{\perp}=1.1$. The blue and red spheres at the vertices represent negative and positive Gaussian curvature.}
    \label{fig:fourfold_pattern}
\end{figure}
The first director pattern we consider is a quadrilateral pattern consisting of four-fold intersections. The starting point is a $3\times 3$ pattern with an unknown director at the $(3,3)$ panel (see Fig.~\ref{fig:fourfold_pattern}(a)). The other directors, creases and folding branches are prescribed. Suppose a reference folding angle $\theta$ is given at a crease of the $(1,1)$ panel.  Due to the one-DoF feature, the folding angles $\gamma_3, \gamma_4$ are determined on the $(3,3)$ panel, as shown in the red frame. Meanwhile, to satisfy the metric compatibility, the tangents $\bft_1(\bfn), \bft_2(\bfn)$ and the sector angles on the $(3,3)$ panel are functions of the unknown director $\bfn$. We will show that $\bfn$ is uniquely determined.

Notice that there are two ways to compute the actuated sector angle $\tilde{\alpha}_1$ for the $(3,3)$ panel (see Fig.~\ref{fig:fourfold_pattern}(a)). On the one hand, the sector angle ${\alpha}_3$ is prescribed. $\alpha_2(\bfn)$and ${\alpha}_4(\bfn)$
are functions of $\bfn$.
Then we can compute the actuated sector angles $\tilde{\alpha}_3$, $\tilde{\alpha}_2(\bfn)$ and $\tilde{\alpha}_4(\bfn)$. Together with the folding angles $\gamma_3$ and $\gamma_4$, the directions of the actuated tangents $\tilde{\bft}_1(\bfn)$ and $\tilde{\bft}_2(\bfn)$ can be calculated. By direct calculation, the angle $\tilde{\alpha}_1$ between $\tilde{\bft}_1(\bfn)$ and $\tilde{\bft}_2(\bfn)$ is 
\begin{align}
\alpha_1^\star({\alpha}_2(\bfn), {\alpha}_3, {\alpha}_4(\bfn), \gamma_3, \gamma_4)  = \arccos[ 
\cos\tilde{\alpha}_2(\bfn) \left( \cos\tilde{\alpha}_3\cos\tilde{\alpha}_4(\bfn) - \cos\gamma_4\sin\tilde{\alpha}_3 \sin\tilde{\alpha}_4(\bfn) \right) \nonumber\\
 - \sin\tilde{\alpha}_2(\bfn) ( 
  \cos\gamma_3 \sin\tilde{\alpha}_3 \cos\tilde{\alpha}_4(\bfn) +  \cos\gamma_3 \cos\gamma_4 \cos\tilde{\alpha}_3 \sin\tilde{\alpha}_4(\bfn) - \sin\tilde{\alpha}_4(\bfn)\sin\gamma_3\sin\gamma_4 
) 
],
\end{align}
where $\tilde{\alpha}_i$ is the actuated sector angle corresponding to $\alpha_i$. On the other hand, the actuated sector angle $\tilde{\alpha}_1$ can be directly calculated from the $(3,3)$ panel, given by
 \beq
\alpha_{def} ({\bft}_1(\bfn),{\bft}_2(\bfn), \bfn) = \arccos\left[\frac{\bfU_{\bfn} \bft_1(\bfn) \cdot \bfU_{\bfn} \bft_2(\bfn)}{|\bfU_{\bfn} \bft_1(\bfn)||\bfU_{\bfn} \bft_2(\bfn)|}\right].
 \eeq
  Solving the equation numerically
\beq
\alpha_1^\star({\alpha}_2(\bfn), {\alpha}_3, {\alpha}_4(\bfn), \gamma_3, \gamma_4) = \alpha_{def} ({\bft}_1(\bfn),{\bft}_2(\bfn), \bfn) \label{eq:marching}
\eeq
yields a (generically unique) solution for $\bfn$, which can be used to develop a marching algorithm for the entire pattern.

To design a quadrilateral pattern, we prescribe the directors and the corresponding tangents at the left and bottom boundaries (panels with indices $(i,1),(i,2),(1,j),(2,j)$, i.e. the domain outside the red dashed lines in Fig.~\ref{fig:fourfold_pattern}(b)). By entering a reference folding angle $\theta$ and applying Eq.~(\ref{eq:marching}), the pattern within the red dashed domain is determined. Once all directors are determined, we fix the tangent branches and compute the tangents. The coordinates of the vertices are then computed by these tangents. 
The difficulty we might have is that the process may break the topology of a quadrilateral mesh. Thus, a more realistic way is to start with a symmetric pattern and then perturb the boundary directors.
Figure~\ref{fig:fourfold_pattern}(b) is an example of irregular quadrilateral patterns. We start with a symmetric pattern with directors along $\bfR_{\bfe_3}(\pi/4) \bfe_1$ and $\bfR_{\bfe_3}(3\pi/4) \bfe_1$. As shown in Fig.~\ref{fig:fourfold_pattern}(b), the symmetric pattern consists of Type A nodes and Type B nodes, which will deform into pyramids and saddles, respectively. We then
perturb the directors at the boundary vertices and compute the entire pattern using Eq.~(\ref{eq:marching}). 
During the marching process, the folding angle branches are fixed by $\delta=+1$ in Eq.~(\ref{eq:gamma1234}), meaning the opposite folding angles have the same sign. The resulting quadrilateral pattern is irregular in both the reference and actuated configurations, and the distribution of Gaussian curvature is non-symmetric. In the deformed configuration, we observe alternating positive and negative Gaussian curvature at the four-fold vertices. We also show the actuated configuration and Gaussian curvature of the orthogonal dual in Fig.~\ref{fig:fourfold_pattern}(c).

{
We may wonder if the structures in Fig.~\ref{fig:fourfold_pattern} are foldable or not.
If we fix $\lambda_{\parallel}=0.9, \lambda_{\perp}=1.1$, these actuated configurations are rigid, because numerically we find that a different reference folding angle leads to a different reference director pattern (except for periodic cases in Fig.~\ref{fig:meta}). However, if we fix the reference folding angle and the reference director pattern, then gradually develop the stretches $\lambda_{\parallel} = \lambda, \lambda_{\perp}=\lambda^{-0.5}$ by decreasing $\lambda$ from 0.95 to 0.65, we observe a continuous compatible evolution (Fig.~\ref{fig:sigma}(a)). This observation relies on the fact that, if we assume the reference pattern is compatible and there exists a compatible deformed configuration for $\lambda_{\parallel, \perp}$,  we can always find a compatible deformed configuration for other $\lambda_{\parallel, \perp}$. Though not rigorously proved, the conjecture is likely to be true based on numerical observations, as also shown in Figs. \ref{fig:meta}(d) and \ref{fig:pattern_34}(d).

The above designs have $\delta=+1$ for the folding branch. Here we provide another example with a different folding branch $\delta$ to show the versatility of our design strategy. The two deformed configurations of the example in Fig.~\ref{fig:sigma}(b) share the same reference director pattern, actuation stretches $\lambda_{\parallel}=0.95, \lambda_{\perp}=1.05$, reference folding angle $\eta=0.8$, but have different folding branches $\delta$. 
The right deformed configuration has all $\delta = +1$, exhibiting shapes similar to those shown in Fig.~\ref{fig:fourfold_pattern}. For the left deformed configuration, we switch the folding branch from $\delta = +1$ to $\delta = -1$ along a row of vertices, leading to sharp geometric changes. Since self-intersection easily occurs in this case, only a $3\times3$ pattern is shown here. It should be noted that, unlike the rigidly and flat-foldable quadrilateral mesh origami in our previous work \cite{feng2020designs}, we have no explicit formula for $\delta$ of the forth node as a function of the geometries and folding branches of the first three, by solving the compatibility equation in Fig.~\ref{fig:fourfold_pattern}(a). However, it is numerically verified that we can always find a solution of $\delta$ for the fourth node as shown in Fig.~\ref{fig:sigma}(b), given the folding branches of the first three ($\delta=+1$ or $\delta=-1$).

\begin{figure}[!ht]
	\center
\includegraphics[width=\textwidth]{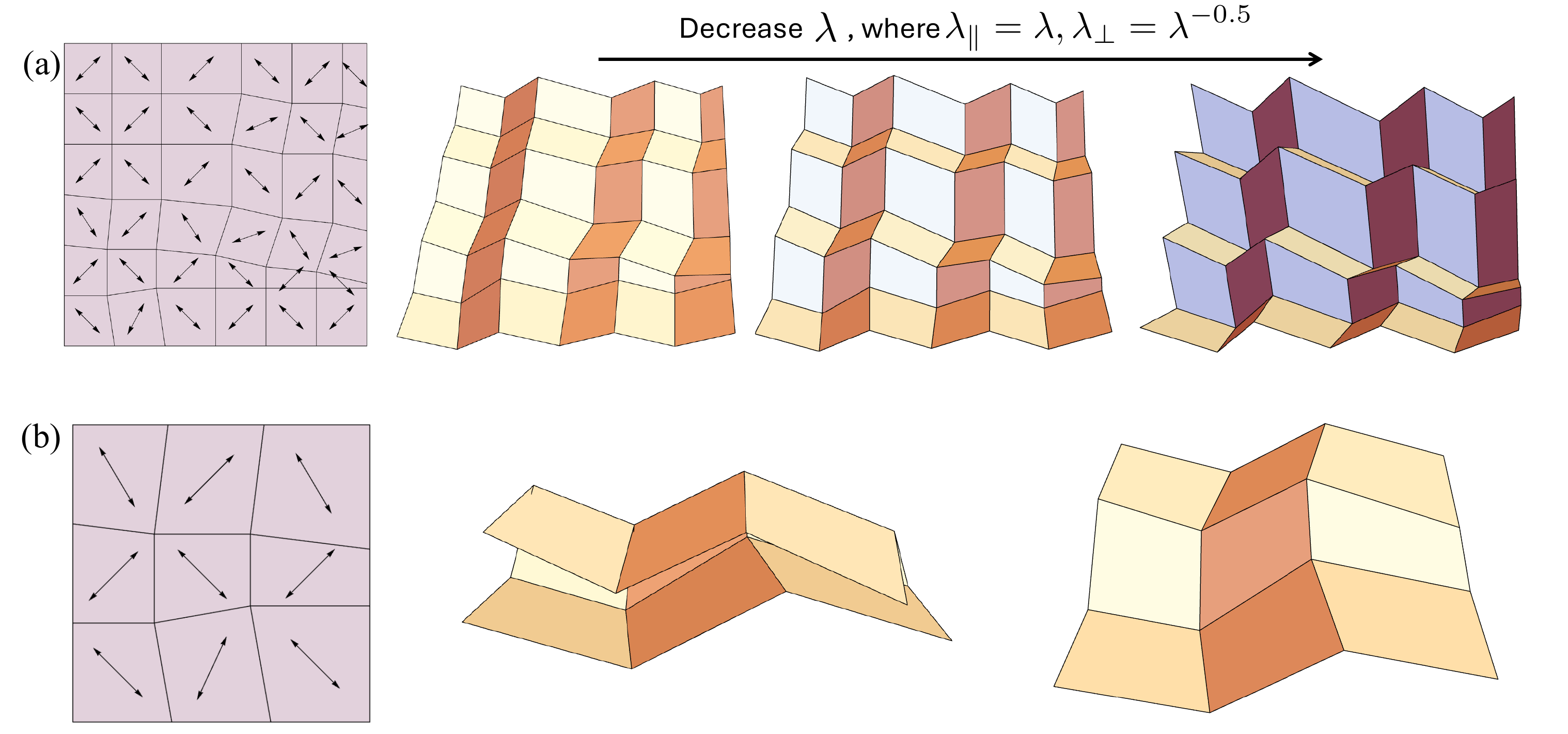}
	\caption{{(a) The evolution of the same director pattern by decreasing $\lambda$ from 0.95 to 0.65 where $\lambda_{\parallel} = \lambda, \lambda_{\perp}=\lambda^{-0.5}$.(b) The two deformed configurations share the same reference director pattern, actuation stretches $\lambda_{\parallel}=0.95, \lambda_{\perp}=1.05$, reference folding angle $\eta=0.8$, but have
    different folding branches $\delta$.}
    } \label{fig:sigma}
\end{figure}

}

{
\subsection{Periodic quadrilateral pattern and active metamaterial}
Harnessing the unique properties of non-Euclidean origami, we may design active metamaterials \cite{xiao2020active,dudek2025shape} that are responsive to external stimuli. A convenient strategy is to use periodic patterns, as shown in Fig.~\ref{fig:meta}. The reference domain is a periodic square pattern, with directors pointing along $\pm\pi/4$ at each panel (Fig.~\ref{fig:meta}(a)). We first study the motion of a single pattern and then try to construct a metamaterial based on it in 3D.  We also notice that the evolved structure is exactly the eggbox origami studied by Paulino in \cite{pratapa2019geometric}. 
\begin{figure}[!ht]
	\center
\includegraphics[width=\textwidth]{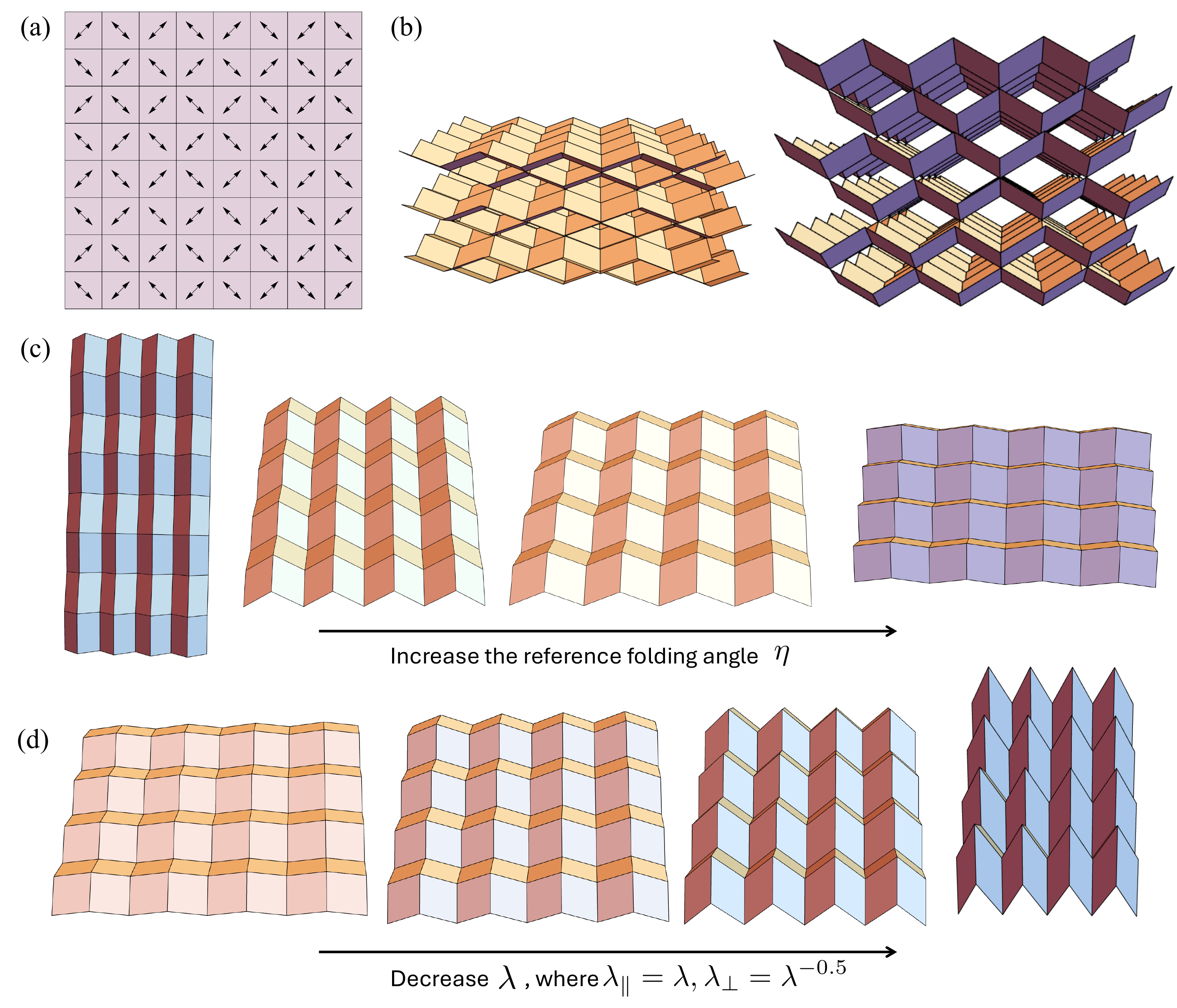}
	\caption{(a) A periodic reference director pattern. (b) Metamaterials constructed by periodic non-Euclidean origami. The joined vertices remain connected during motion. (c)-(d) Two deformation modes regarding rigid folding and active evolution respectively. (c) Fix the stretches $\lambda_{\parallel}, \lambda_{\perp}$ and change the reference folding angle $\eta$. (d) Fix the reference folding angle $\eta$ and change the stretches $\lambda_{\parallel} = \lambda, \lambda_{\perp}=\lambda^{-0.5}$ from $\lambda=0.95$ to $\lambda=0.5$.
    } \label{fig:meta}
\end{figure}

The non-Euclidean origami has two modes of motion. First we assume the director pattern (each panel) is fully evolved by $\lambda_{\parallel}=0.9$ and $\lambda_{\perp}=1.1$ in Fig.~\ref{fig:meta}(c), and then change the folding angle $\eta$ from $0.3$ to $1.5$. Owing to the symmetry, the origami is rigidly foldable, exhibiting an overall shape change from rectangle to square and to rectangle shape as shown in Fig.~\ref{fig:meta}(c). The other mode is that we can fix the folding angle $\eta=1$ and evolve the stretches $\lambda_{\parallel} = \lambda, \lambda_{\perp}=\lambda^{-0.5}$ from $\lambda=0.95$ to $\lambda=0.5$. Unlike the first mode, the overall pattern evolves through large-square, small-square, and rectangular shapes.
The second mode mimics the shape evolution during the phase transformation of liquid crystal elastomer sheets.

Since the origami pattern remains compatible and periodic during the two modes of motion, we could stack the actuated pattern in multilayers by joining corresponding nodes periodically (Fig.~\ref{fig:meta}(b)). The stacked non-Euclidean origami forms an active metamaterial that can undergo shape change by folding it (mode 1) and by stimulating it (mode 2). The latter would equip our designed active metamaterial with better properties compared with classic passive metamaterial, since many physical properties related to shapes (e.g. band gaps, damping, Poisson's ratio) can be modulated upon external stimuli in our case. A detailed study of the active metamaterial is out of the scope of this paper, and we reserve it for future work.

}
\subsection{Director patterns combining three-fold and four-fold intersections}
\begin{figure}[!ht]
    \centering
    \includegraphics[width=\linewidth]{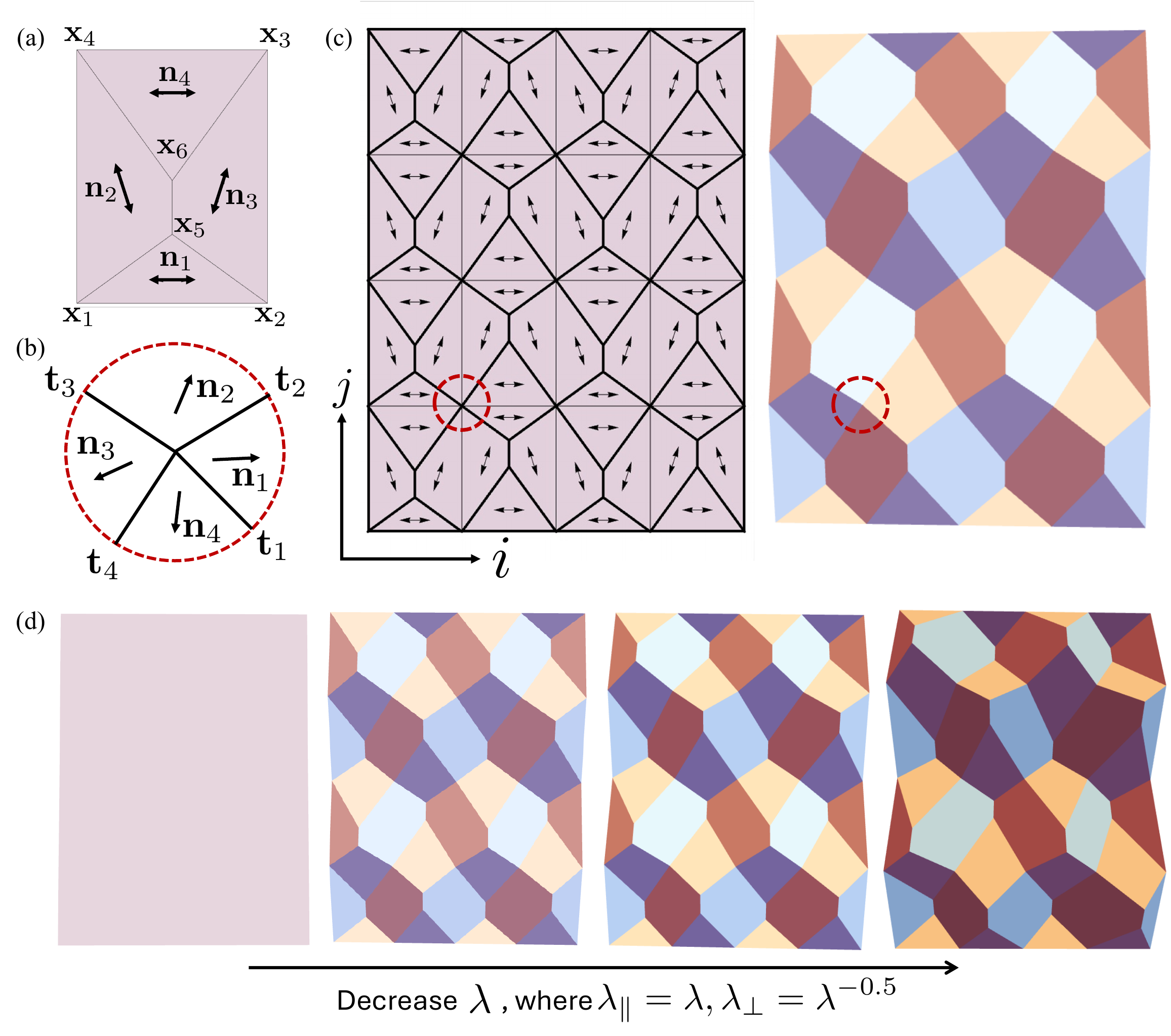}
    \caption{(a) Notation of the unit cell. (b) Notation of the four-fold intersection in the pattern. (c) The designed pattern composed of three-fold and four-fold intersections. The pattern is compatible in both the reference and actuated configurations. (d) The evolution of the structure by changing the stretches $\lambda_{\parallel} = \lambda, \lambda_{\perp}=\lambda^{-0.5}$ from $\lambda=1$ to $\lambda=0.5$.}
    \label{fig:pattern_34}
\end{figure}

The second pattern we design is a combination of three-fold and four-fold vertices as shown in Fig.~\ref{fig:pattern_34}. The unit cell is composed of two connected three-fold intersections inherited from the pattern in Fig.~\ref{fig:angle_change}, with the notation shown in Fig.~\ref{fig:pattern_34}(a). The directors on
adjacent unit cells satisfy $\bfn_3(i,j) = \bfn_2(i+1,j)$ and $\bfn_4(i,j)=\bfn_1(i,j+1)$. Then the boundaries of unit cells in Fig.~\ref{fig:pattern_34}(c) (i.e. thin lines) are not folding lines, and the two adjacent panels on two unit cells form a single panel. 
The actual folding lines in Fig.~\ref{fig:pattern_34}(c) are thick lines.
We then deform the reference pattern by popping the unit cells up or down, indicated by $\delta=+1$ or $\delta=-1$ respectively. It is noted that if the distribution of $\delta$ satisfies 
\begin{align}
    \begin{cases}
        \delta(i,j) = +1, \quad \text{} i+j \text{ odd} \\
        \delta(i,j) = -1, \quad \text{} i+j \text{ even} 
    \end{cases},
\end{align}
the deformed pattern is compatible, in particular at the four-fold intersection. 
Figure~\ref{fig:pattern_34}(b) shows a generic four-fold intersection in the pattern. We will show that the deformed four-fold intersection is compatible if the folding angles are computed from the three-fold intersections.
For an arbitrary four-fold intersection in the pattern (Fig.~\ref{fig:pattern_34}(b)), we assume the directors are given by $\bfn_i = \bfR_{\bfe_3}(\theta_i)\bfe_1, 0<\theta_1<\theta_2<\theta_3<\theta_4<2\pi$, which satisfy the topological condition of creases. We compute the folding angle $\gamma_i$ using Eq.~(\ref{eq:thm_cosgamma}) from the three-fold kinematics
	\beq
	\gamma_i =(-1)^i\arccos\left[r^2 + (-1 + r^2) \cos(\theta_i - \theta_{i-1})\right], ~i=1,2,3,4.
	\eeq
Here we have already substituted the tangent branch $\hat{\sigma}_i$ for this particular pattern and $(-1)^i$ accounts for the pop-up or pop-down of the three-fold nodes.
That is, the folding angles are inherited from the three-fold intersections, and the folds of the actuated four-fold nodes have alternating mountains and valleys folding lines.
Then we perform the numerical validation for the compatibility equation at the four-fold vertex:
\beq
error = \max_{0<\theta_1<\theta_2<\theta_3<\theta_4<2\pi} |\bfR_2(\gamma_2) \bfR_3(\gamma_3) \bfR_4(\gamma_4) \tilde{\bft}_5 - \tilde{\bft}_1|
\eeq
for any $\theta_i$ satisfying the topological condition $\bft_i \times \bft_{i+1} \cdot \bfe_3 = 1$,
where $\bfR_i$ is a rotation tensor with rotation axis $\tilde{\bft}_i$ for the intermediate configuration. 
The numerical result given by the global minimization NMaximize in Mathematica shows that $error=0$ in the domain $0<\theta_1<\theta_2<\theta_3<\theta_4<2\pi$, indicating that the four-fold vertex is compatible in our setting. Since compatibility holds for arbitrary four-fold nodes, we can successively stack different unit cells along the i- and j-directions to construct both the reference pattern and the actuated configuration, as shown in Fig.~\ref{fig:pattern_34}(c). In the actuated state, the configuration exhibits positive Gaussian curvature at the three-fold vertices and negative Gaussian curvature at the four-fold vertices. {Notably, the design space for the director fields is quite large; the primary constraint is the preservation of the pattern’s topological structure. We also notice that the actuated configuration can be evolved in a compatible way by reducing $\lambda=1$ to $\lambda=0.5$ in the stretches $\lambda_{\parallel} = \lambda, \lambda_{\perp}=\lambda^{-0.5}$,  as discussed in the previous section.}

\section{Discussion} \label{sec:discussion}
In this work, we have presented a geometric framework for designing complex active origami patterns, where each panel is equipped with a prescribed metric change. This metric change, induced by local actuation, leads to nontrivial Gaussian curvature at the vertices, rendering the origami non-Euclidean. Such systems can be physically realized using two-dimensional liquid crystal elastomer (LCE) sheets, where a patterned director field encodes the spatially varying metric change. By assuming that each panel remains flat and rigid after actuation, and by neglecting crease energy, we reduce the design challenge to a purely geometric problem: how to construct a piecewise constant director field that ensures compatibility in both the reference and actuated configurations.

To tackle this problem, we first analyze the fundamental condition of metric compatibility between adjacent director fields. { 
We then study the compatible director fields for a general $n$-fold vertex, characterized into the generic case and the degenerate case. We prove that for the generic case there are two and only two compatible director fields which are orthogonal to each other (orthogonal duals), while for the degenerate case, there is a continuous family of compatible director fields. The Gaussian curvature of the actuated vertex is also computed.}
We develop a rigorous continuum mechanics framework to study the kinematics of isolated three-fold and four-fold vertices. Building on this vertex-level understanding, we construct two classes of large-scale director patterns that are compatible both before and after actuation. In other words, the directors satisfy metric compatibility at each crease, and the resulting global deformation is geometrically compatible.
Although the general compatibility problem remains open ({for example, a three-fold compatible director pattern that develops negative Gaussian curvature is always incompatible in the actuated configuration}), our two proposed designs demonstrate broad and flexible design spaces. The first is a quadrilateral director pattern with free directors on the two boundaries. Using a marching algorithm and a rigidity theorem, the director information propagates to interior nodes, producing a compatible actuated configuration characterized by alternating positive and negative Gaussian curvature at the vertices. The second design combines three-fold and four-fold intersections, leveraging a folding angle theorem that shows the folding angle at a three-fold node depends solely on its adjacent directors. In this configuration, positive Gaussian curvature arises at the three-fold vertices and negative curvature at the four-fold vertices after actuation. {On the application side, we propose a design strategy for active metamaterials based on periodic quadrilateral director patterns, which reveals two modes of motion by folding and stimulating.}

Looking ahead, several promising directions merit further exploration. On the mechanics side, one could investigate the effective behavior of these patterns under external loads. For instance, drawing inspiration from effective plate theories developed for origami based on bar-and-hinge models \cite{XU2024105832,xu2025modelingcomputationeffectiveelastic,filipov2017bar}, it is natural to ask whether an analogous plate theory for non-Euclidean origami can be formulated, and what  the embedded Gaussian curvature would add to the system.
On the application side, our designs could serve as building blocks for energy-absorbing layers, deployable structures, or load-bearing elements, where the induced Gaussian curvature provides additional strength and stiffness, as seen in recent studies of non-Euclidean curved-fold actuators \cite{duffy2021shape,feng2022interfacial,feng2024geometry,zhai2020situ}.
Finally, from a geometric perspective, the complete solution of the general compatibility problem—determining whether a given compatible piecewise-constant director field leads to compatible actuated configurations—remains open. Addressing this challenge may require an augmented non-isometric, piecewise-constant analog of the Nash embedding theorem \cite{nash1954c} or a non-isometric extension of the crumpled paper theory \cite{conti2008confining}.

\section*{Acknowledgments} 
This work is based on discussions with Prof. Paul Plucinsky at USC, while we were working together at the University of Minnesota. We also thank Mr. Xi SU from PKU for providing the Householder matrix approach for the proof of Theorm 3.1.
F.F. acknowledges the financial support from the National Natural Science Foundation of China (Grant No. 12472061).
\bibliographystyle{unsrt}
\bibliography{origami}

\end{document}